\documentclass{amsart}
\usepackage{amsmath}
\usepackage{amssymb, colordvi}
\usepackage{multirow}
\usepackage{mathrsfs}
\usepackage{xcolor}
\usepackage[theorems]{tcolorbox}
\usepackage[authoryear,round]{natbib}
\usepackage[all]{xy}
\usepackage{tikz}
\usepackage{pgf}
\usepackage{pgffor}
\usetikzlibrary{arrows,positioning,matrix,decorations.pathmorphing,calc,fadings,decorations.pathreplacing} 
\usepackage{hyperref}
\usepackage{todonotes} 
\usepackage{enumerate}
\usepackage[version=4]{mhchem}
\usepackage[all]{xy}
\usepackage{tikz-cd}

\usepackage{makecell}

\numberwithin{equation}{section}
\newtheorem{theorem}{Theorem}
\numberwithin{theorem}{section}
\newtheorem{proposition}[theorem]{Proposition}
\newtheorem{lemma}[theorem]{Lemma}
\newtheorem{corollary}[theorem]{Corollary}

\theoremstyle{definition}
\newtheorem{example}[theorem]{Example}

\newtheorem{remark}[theorem]{Remark}
\newtheorem{definition}[theorem]{Definition}

\newcounter{FNC}[page]
\def\fauxfootnote#1{{\addtocounter{FNC}{2}$^\fnsymbol{FNC}$%
     \let\thefootnote\relax\footnotetext{$^\fnsymbol{FNC}$#1}}}

\newcommand{\R}{\mathbb{R}}

\newcommand{\Q}{\mathbb{Q}}

\newcommand{\Z}{\mathbb{Z}}

\newcommand{\Sp}{\mathscr{S}}

\newcommand{\uri}{\rightarrow_\circ}

\begin{document}

\title[Reduced vs extended networks]{Multistationarity questions in reduced vs extended biochemical networks} 

\author[A. Dickenstein]{Alicia Dickenstein}
\address{Dto.\ de Matem\'atica, FCEN, Universidad de Buenos Aires, and IMAS (UBA-CONICET), Ciudad Universitaria, Pab.\ I, 
C1428EGA Buenos Aires, Argentina}
\email{alidick@dm.uba.ar}
\urladdr{http://mate.dm.uba.ar/~alidick}

\author[M. Giaroli]{Magal\'i Giaroli}
\address{Dto.\ de Matem\'atica, FCEN, Universidad de Buenos Aires, Ciudad Universitaria, Pab.\ I, 
C1428EGA Buenos Aires, Argentina}
\email{mgiaroli@dm.uba.ar}

\author[M. P\'erez Mill\'an]{Mercedes P\'erez Mill\'an}
\address{Dto.\ de Matem\'atica, FCEN, Universidad de Buenos Aires, and IMAS (UBA-CONICET), Ciudad Universitaria, Pab.\ I, 
C1428EGA Buenos Aires, Argentina}
\email{mpmillan@dm.uba.ar}
\urladdr{http://cms.dm.uba.ar/Members/mpmillan}

\author[R. Rischter]{Rick Rischter}
\address{\sc Rick Rischter\\
Universidade Federal de Itajub\'a (UNIFEI)\\ 
Av. BPS 1303, Bairro Pinheirinho\\ 
37500-903, Itajub\'a, Minas Gerais\\ 
Brazil}
\email{rischter@unifei.edu.br}
\urladdr{http://w3.impa.br/~rischter/}

\date{}

\begin{abstract}
We address several questions in  reduced versus extended networks via the elimination or addition of intermediate complexes in the framework of chemical reaction networks with mass-action kinetics. We clarify and extend advances in the literature concerning multistationarity in this context, mainly from~\cite{fw13,SFeliu,messi,DPMST}. We establish  general results about MESSI systems, which we use to compute the circuits of multistationarity for significant biochemical networks.
 
\end{abstract}
\maketitle

\section{Introduction}

The first systematic study of the role of intermediate complexes in the setting of biochemical reaction networks  and the extension of multistationarity results in this context was introduced in the thoughtful paper~\cite{fw13}. We clarify and extend results from this paper and also from~\cite{Focm,messi,DPMST,SFeliu}. In particular, we give explicit {\em open conditions} on the parameters that ensure the lifting of multistationarity.  We rephrase the notion of circuits of multistationarity introduced in~\cite{SFeliu} and we compute the minimal additions of intermediates that allow for multistationarity in different variations of the ERK pathway,  a cascade of protein reactions in the cell that communicates a signal from a receptor on the surface of the cell to the nucleus~\cite{PS18}. 
The MESSI networks introduced in~\cite{messi} include the ERK pathway and are abundant in the literature. They give us a  unified framework to study different properties for many classes of interesting networks that can be easily determined with graphical tools. We concentrate on conditions to decide that the associated systems are toric~\cite{PMDSC} by means of $\R$-linear operations and we prove that many interesting networks as the ERK pathway have this property.
In order to introduce the necessary definitions and state our results, we first recall the basic setup of chemical reaction networks and how they give rise to autonomous dynamical systems under mass-action kinetics. 
For a more comprehensive overview of Chemical Reaction Networks we refer to \cite{alicia, notices}.  

\smallskip

Given a set of $s$ chemical species $X_1,X_2,\dots, X_s$, a {\em chemical reaction network} on this set of species is a finite directed graph $G$ whose vertices are indicated by complexes (linear combinations of the species with nonnegative integer coefficients) and whose edges are labeled by parameters (reaction rate constants). The set of species is denoted by $\Sp_G$, the  vertex set by $\mathscr{C}_G$, the edge set by $\,\mathcal{R}_G$, and the edge labels by $\kappa=\{\kappa_{yy'}\} \in \R_{>0}^{\# \mathcal{R}_G}$. Here, $\kappa_{yy'}$ is the reaction rate constant associated to the reaction $(y,y')\in \mathcal{R}_G$, which is in turn denoted by $y\rightarrow y'$.

The unknowns $x_1,x_2,\ldots,x_s$ represent, respectively, the concentrations of the species $X_1, \dots,X_s$ in the network, and we regard them as functions of time $t$. Under mass-action kinetics, the chemical reaction network $G$ defines the following chemical reaction dynamical system for $x=(x_1,\dots,x_s)$:
\begin{equation}
\label{eq:CRN}
\dot{x}~:=~f_\kappa(x)~=~\underset{y\to y'}{\sum} \kappa_{yy'} \,  x^y \, (y'-y),
\end{equation}
where $x^y=x_1^{y_1}\cdots x_s^{y_s}$. The right-hand side of each differential equation $\dot{x}_\ell$ is a polynomial $f_{\kappa,\ell}(x)$, in the variables $x_1,\dots, x_s$ with positive coefficients  $\kappa$.

A {\em steady state} of the system is a nonnegative concentration vector $x^* \in \mathbb{R}_{\geq 0}^s$ at which the ODEs~\eqref{eq:CRN}  vanish, i.e. $f_\kappa (x^*) =0$. 
We distinguish between {\em positive steady states} $x ^* \in \mathbb{R}^s_{> 0}$ and {\em boundary steady states} 
$x^*\in {\mathbb R}_{\geq 0}^s\backslash {\mathbb R}_{>0}^s$. Both the positive orthant $\mathbb R_{>0}^s$ and its closure $\mathbb R^s_{\ge 0}$ are forward-invariant for the dynamics. 

The linear subspace spanned by the reaction vectors $\mathcal{S} = \{ y' -y\, : \, y \to y' \}$ is called the {\em stoichiometric subspace}. 
Vectors in $\mathcal{S}^\perp$ yield the equations of $x^0+\mathcal{S}$  for any $x^0 \in \R^s$, and give rise to linear \emph{conservation relations} of the system. If $d=s-\dim(\mathcal{S})$, and $W$ is a row-reduced $d\times s$-matrix whose   rows form a basis of $\mathcal{S}^{\perp}$, then $W$ is called a {\em conservation-law matrix} of $G$, and $W \dot x = Wf_\kappa(x)=0$. 
Thus, a trajectory $x(t)$ beginning at a nonnegative vector $x(0)=x^0 \in \mathbb{R}^s_{\geq 0}$ remains, for all $t$ in any interval containing $0$ where $x$ is defined, in the following {\em stoichiometric compatibility class} with respect to the {\em linear conservation vector} $c:= W x^0$: 
\[
\mathcal{S}_c~:=~ \{x\in {\mathbb R}_{\geq 0}^s \mid Wx=c\}.
\]
In particular, $\mathcal{S}_c$ is also forward-invariant with respect to the dynamics~\eqref{eq:CRN}. 

The system is said to be \emph{conservative} if there exists a positive vector in $\mathcal{S}^\perp$. In this case, all the stoichiometric compatibility classes are compact and trajectories  around $0$ are defined for any positive $t >0$. 
We say that the system {\em has the capacity for multistationarity} if there exists a choice of rate constants $\kappa$ such that there are two or more positive steady states in one stoichiometric compatibility class, called stoichiometrically compatible (scpss). On the other hand, if for any choice of rate constants there is at most one positive steady state in each stoichiometric compatibility class, the system is said to be \emph{monostationary}.

\medskip

In this paper, we give the fundamental notions of {\em intermediate species and complexes} in Definition~\ref{def:int} and of {\em reduced and extended networks by intermediates} in Definition~\ref{def:redext}.
The relations between the conservation laws in a given network $G$ and a reduced network $G_{red}$ as well as the relation between the corresponding rate constants as in~\eqref{eq:T} and steady states were studied in~\cite{fw13}. We summarize their results in Section~\ref{sec:intermediates}. We introduce the notion of non-confluent networks in Definition~\ref{def:nc}
and we prove in Proposition~\ref{prop:bss} that a system does not have relevant boundary steady states if and only if any of its non-confluent extensions has this property.

\smallskip

The lifting of steady states to an extended network was first studied in Theorem~5.1 in~\cite{fw13}. The main result in the recent paper~\cite{Murad23} studies the lifting of multistationarity and of periodic orbits. Feliu and Wiuf proved that if we consider a network $G$ with rate constants $\kappa^0$ and $\tau^0=T(\kappa^0)$ are the corresponding rate constants of the reduced network $G_{red}$, if $G_{red}$ has $m$ non-degenerate scpss then there exist infinitely many rate constants $\kappa$ such that $\tau^0 = T(\kappa)$ and the extended system $G$ with these rate constants has at least $m$ non-degenerate scpss. Note that this doesn't say that the original system $G$ with rate constants $\kappa^0$ is multistationary. Using the same ideas of their proof, we improve their result in Theorems~\ref{thm:IFT} and~\ref{thm:IFT2} by identifying (a finite number of) explicit rational functions of $\kappa$ which ensure multistationary of $G$ when they are small and in many common networks, if the explicit rate constants of reactions with source in an intermediate complex are big enough.

\smallskip

In Section~\ref{sec:circuits} we recall and simplify results in~\cite{SFeliu}. They introduced the notion of {\em circuits of multistationarity} (see Definition~\ref{def:circuits}): given a monostationary network $G'$ with steady states defined by binomial equations, a circuit of multistationarity is a subset of the complexes of $G'$ such that the addition of intermediates from these complexes gives raise to a multistationary system (for some choice of matching rate constants), which is minimal with respect to inclusion.  In fact, we need to specify the meaning of {\em defined by binomial equations}. In~\cite{SFeliu} they present a very general setting of {\em complete binomial networks}. We introduce the more restrictive but still quite general notion of networks {\em linearly equivalent to a binomial network} (called {\em lebn} networks, see Definition~\ref{def:lebn}).
We prove  in Proposition~\ref{prop:CK} that one can check this condition computing a reduced row echelon form of the matrix of the system an in this setting we state Theorem~\ref{th:Bpoly} (that also holds with more general hypotheses).
All examples in~\cite{SFeliu} are lebn as well as the important networks we study in the subsequent sections.

\smallskip

Section~\ref{sec:ERK} is concentrated on the computation of the circuits of multistationarity for the {\em mitogen-activated protein kinase (MAPK)} pathway, commonly known as the MAPK/ERK signaling pathway. The acronym ERK refers to {\em Extracellular Signal-Regulated Kinase}. The initiation of this signaling mechanism occurs when an external stimulus, such as a growth factor, binds to a designated receptor on the cell's surface, followed by a {\em cascade} of enzymatic activations within the cell.  Given the high number of variables and parameters, we implemented the previous results in a computer algebra system.
We refer the interested reader to the section on ``Signaling through Enzyme-Linked Cell-Surface Receptors" in Chapter 15 in~\cite{alberts2002} for a biochemical understanding of the ERK pathway.
 
 \smallskip
 
In  Section~\ref{sec:Messi} we recall the notion of MESSI systems introduced in~\cite{messi}. We give in Theorem~\ref{th:ylebn} and Theorem~\ref{th:monoy} simple combinatorial conditions on the associated digraphs of a MESSI network that ensure that a system  is monostationary and lebn.  Most common biological networks are MESSI, so our results have wide applicability. In particular, the ERK pathway has a MESSI structure as well as the mixed sequential distributive-processive phosphorylation mechanisms, for any number of phosphorylation sites. We compute the circuits of multistationarity of these networks in Section~\ref{sec:sequential}.

\section{Intermediates} \label{sec:intermediates}
In this section, we give in~\S~\ref{ssec:red&ext} the definition of intermediate species and complexes and the notions of extended and reduced networks via intermediates. In \S~\ref{ssec:conslaws} and \S~\ref{ssec:rate_constants} we highlight in a concise way the results in~\cite{fw13} concerning the lifting of conservation laws and the reduction of rate constants. Then, in \S~\ref{ssec:bss}, we introduce the notion of non-confluent extensions by intermediates and we prove Proposition~\ref{prop:bss} about the lifting and reduction of relevant boundary steady states for this kind of extensions. We finally identify in Theorems~\ref{thm:IFT} and~\ref{thm:IFT2} explicit relative open sets such that multistationarity is lifted to the mass-action systems of all the extended networks with rate constants lying in these open sets. In particular, we show that for standard enzymatic networks, this is the case when the catalytic rate constants are big (see Remark~\ref{rem:unasola}).

\subsection{Reduced and extended networks}\label{ssec:red&ext}

We consider a reaction network $G$ and a subset of species $\mathcal{I}=\{U_1,U_2,\dots,U_p\} \subset \Sp_G$ such that for any $i =1, \dots, p$, the only complex that involves species $U_i$ is $U_i$.

\begin{definition}\label{def:uri}
 We say that complex $y$ reacts to complex $y'$  via $\mathcal{I}$ if either $y\to y'$ or there exists a path of reactions from $y$ to $y'$ only through complexes $U_j$. This is denoted  by $y \uri y'$.
\end{definition}

We now define the meaning of intermediate species and complexes.

\begin{definition}[Intermediate and core] \label{def:int}
Let $G$ and $\mathcal I$ be as above. A species $U_i \in \mathcal I$ is called an {\em intermediate species} if there is a sequence of reactions $y_{j} \uri {U_i}\uri y_{k}$, with $y_{j},y_{k}$ complexes that only involve species in $\Sp_G\setminus \mathcal{I}$. When all species in $\mathcal I$ are intermediate, we say that $\Sp_G \setminus \mathcal{I} = \{X_1,\dots, X_n\}$ is the set of \emph{core species}. The complexes  $U_i$ are called  \emph{intermediate complexes} and the complexes not involving intermediate species are called \emph{core complexes}.
 \end{definition}

Reduced and extended networks by intermediates are defined as follows.

\begin{definition}[Reduced and extended networks]\label{def:redext}
 Consider a reaction network $G$ and a subset of intermediate species $\mathcal{I}=\{U_1,U_2,\dots,U_p\} \subset \Sp_G$. 
 The associated reduced network $G_{red,\mathcal{I} }$ is obtained from $G$ by removing the intermediate species in $\mathcal{I}$. 
 The set of species of $G_{red,\mathcal{I}}$ is $\{X_1, \dots, X_n\}$. The complexes of $G_{red,\mathcal{I}}$ are the complexes of  $G$ which are not an intermediate complex  and the set of reactions of $G_{red,\mathcal{I}}$ is obtained from the set of reactions of $G$ by collapsing any sequence $y_{j}\uri y_k$ to the reaction $y_{j}\rightarrow{y_k}$. We will omit the subindex $\mathcal{I}$ when the set of intermediate species is clear from the context. Reciprocally, we say that $G$ is an extended network of $G'= G_{red, \mathcal I}$ via the addition of the intermediate species in $\mathcal I$ and we write $G = G'_{ext, \mathcal I}$.
\end{definition}

\begin{example}\label{ex:int}
Let $G$ be the network  with core complexes $y_1$, $y_2$ and $y_3$, and intermediate complexes $U_1, U_2, U_3$ depicted in Figure~\ref{fig:GG}. 
We show the associated reduced network $G_{red}$ on the right.
\begin{center}
\begin{figure}
 \begin{tabular}{c@{\hskip 1in}c}
$\xymatrix@!=0.4pc{& {y_2}\ar_{\ \kappa_3 \ }@<0.3ex>[d] & &\\ & {U_1}\ar_{\ \kappa_4\ }@<0.3ex>[u]\ar^{\kappa_5}[rd]& &\\ 
 {y_1}\ar^{\kappa_1}[ru]\ar_{\kappa_2}[rd]&& {U_3}\ar^{\kappa_7}[r]& {y_3}\\ & {U_2}\ar_{\kappa_6}[ru] & & }$ &
$\xymatrix{& {y_2}\ar^{\tau_2}[rd] & \\ {y_1}\ar^{\tau_1}[ru]\ar^{\tau_3}[rr]&& {y_3}.}$
\end{tabular}
\caption{The networks $G$ (on the left) and the corresponding network $G_{red}$ (on the right) } \label{fig:GG}
\end{figure}
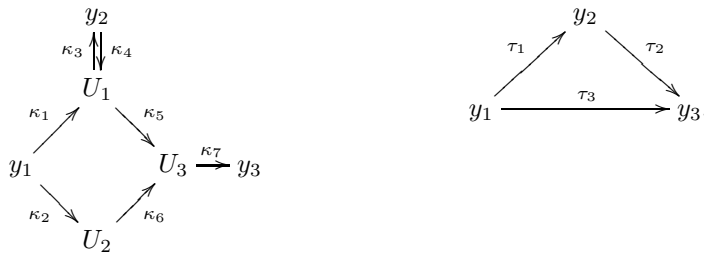
\end{center}
\end{example}

\subsection{Conservation laws}\label{ssec:conslaws}

The conservation laws in $G$ are in one-to-one correspondence with the conservation laws in $G_{red}$. This is detailed in Theorem~2.1 (and also Lemmas~1 and~2 of the ESM) in~\cite{fw13}, which we now briefly recall. Note that if a system is conservative,  by picking any basis of the conservation relations and adding to the linear forms in this basis a sufficiently high multiple of any positive vector in the orthogonal of the stoichiometric subspace $\mathcal S$, we can  assume that all linear forms in the basis have positive coefficients.

We denote as before  $\mathcal{S}^\perp$ the orthogonal of the stoichiometric subspace associated to $G$ and we let $\mathcal{S}_{red}^\perp$ be the corresponding subspace for $G_{red}$. These two linear subspaces in $\R^{n+p}$ and $\R^n$, respectively, are in bijection via the projection onto the first $n$ coordinates. The inverse linear mapping is given as follows: any  linear {conservation relation} $\ell$ of  $G_{red}$ is lifted to the linear conservation relation $\bar{\ell}$ of $G$ defined by:
\begin{equation}\label{eq:cons_laws_G}
 \bar{\ell}(x,u)= \ell(x+\sum_{k=1}^p u_k y^{(k)}) = \ell(x) + \sum_{k=1}^p \ell(y^{(k)}) \,  u_k,
\end{equation}
where $y^{(k)}$ is any choice of a core complex in the same connected component as $U_k$. The following lemma is a straightforward consequence of this description.

\begin{lemma} \label{lem:conservative} 
Let $G$ be a chemical reaction network with set of intermediate species $\mathcal{I}$ and let $G_{red,\mathcal{I} }$ be the associated reduced network. Then, the associated chemical reaction system of $G_{red,\mathcal{I} }$ is conservative if and only if the associated system of $G$ is conservative.
\end{lemma} 

\subsection{Steady states and rate constants}\label{ssec:rate_constants}
We first need to briefly introduce some facts related to the Laplacian $\mathcal{L}(G)$ of a digraph $G$. Recall that a spanning tree of a digraph is a subgraph that contains all the vertices, is connected and acyclic as an undirected graph. An $i$-tree of a graph is a spanning tree where the vertex $i$ is its unique sink (that is, the only node with outdegree zero). When $G$ is strongly connected (i.e. for any ordered pair of vertices of $G$ there is a directed path from the first vertex to the second one), the kernel of $\mathcal{L}(G)$ has dimension one and there is a known generator $\rho(G)$, where the $i$-th coordinate equals:
\begin{equation}\label{eq:rhoi}
\rho_i(G)=\underset{\mathcal{T}\; an \; i-tree}{\sum}\pi(\mathcal{T}),
\end{equation}
where $\pi(\mathcal{T})$ is the product of the labels of all the edges of $\mathcal{T}$.
We refer the reader to~\cite{MiGu13,tutte} for a detailed account.

\medskip

We now recall the relation between the steady states of the mass-action kinetics system associated to a given network with intermediates and the steady states of the corresponding reduced system.
Consider the network $G$ with reaction rate constants $\kappa$. By Theorem 3.1~in \cite{fw13} we have an expression of the concentration of the intermediates at steady state in terms of the reaction rate constants $\kappa$ and the concentration of the core species.
The system of differential equations $\dot{u_i}=0$, for all intermediates $U_i$, $i=1,\dots,p$, is linear on the $u_i's$, and the concentration $u_i$ at steady state can be written as follows in terms of the concentrations of the core species:
\begin{equation}\label{eq:concentracioninter}
u_i=\sum_{y\in{\mathscr{C}_{G_{red}}}}{\mu_{i,y}(\kappa) \, x^y}.
\end{equation}
Here, $y$ denotes a core complex and it holds that $\mu_{i,y}\neq 0$ if and only if $y\uri U_i$. In this case,  ${\mu_{i,y}(\kappa)}$ is a nonnegative rational function on the reaction rate constants $\kappa$ with homogeneous numerator and denominator. 
In fact, in the proof of Theorem~3.1 of~\cite{fw13} it is shown how to obtain $\mu_{i,y}$ from a graphical procedure that we briefly recall.
For a fixed core complex $y$ consider the digraph $G^y$ whith node set $\{*\}$ and all the intermediates $U_i$ such that $y\uri U_i$, and edges $U_i\overset{\kappa_{ij}}{\longrightarrow}U_j$, $*\overset{\kappa_{yU_j}}{\longrightarrow}U_j$ if the corresponding constant is nonzero, and $U_i\overset{\sum_{U_i\to y'}\kappa_{U_iy'}}{\longrightarrow}*$ (i.e. if there are several core complexes to which $U_i$ reacts, the edges are collapsed and the label equals the sum of the labels of the corresponding collapsed edges). It is easy to see that $G^y$ is strongly connected. Consider the generator $\rho(G^y)$ of the kernel of $\mathcal{L}(G^y)$. If $\rho_i$ denotes the entry of $\rho(G^y)$ that corresponds to $U_i$ and $\rho_*$ the entry of $\rho(G^y)$ that corresponds to $*$, then $\mu_{i,y}=\frac{\rho_i}{\rho_*}$.
Thus, the denominator does not vanish over the positive orthant. 
We   show an explicit computation in the following example.

\begin{example}[Example~\ref{ex:int}, continued] \label{ex:int_cont}
 Consider the reaction networks $G$ and $G_{red}$ in Figure~\ref{fig:GG}.  
 Here we explain how to  obtain $\mu_{1,y_1}$.

\noindent\begin{minipage}{0.75\textwidth}
 Consider $G^{y_1}$, the digraph with node set $\mathcal{I}\cup\{*\}$ and labeled edges, constructed from $G$ in Figure~\ref{fig:GG} as shown on the right. Call $\rho$ the generator of the kernel of $\mathcal{L}(G^{y_1})$, then:
 $$\mu_{1,y_1}(\kappa)=\frac{\rho_1}{\rho_{*}}=\frac{\kappa_1\kappa_6\kappa_7}{\kappa_3\kappa_6\kappa_7 + \kappa_5\kappa_6\kappa_7}=\frac{\kappa_1}{\kappa_3+\kappa_5}.$$
\end{minipage}
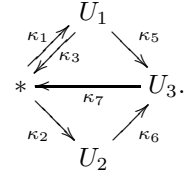
\begin{minipage}{0.25\textwidth}
 \[\xymatrix@!=0.3pc{& {U_1}\ar^{\kappa_5}[rd] \ar_{\kappa_1}[ld]& \\ 
 {*}\ar_{\kappa_3}@<0.8ex>[ru]\ar_{\kappa_2}[rd]&& {U_3.}\ar^{\kappa_7}[ll]\\ 
 & {U_2}\ar_{\kappa_6}[ru] & }\]
\end{minipage}

\end{example}

\medskip

Consider a network $G$ with reaction rate constants $\kappa$.
Denote by $\varphi_i(x) =\sum_{y\in{\mathscr{C}_{G_{red}}}}{\mu_{i,y}(\kappa) x^y}$ for $i=1, \dots, p$, and $\varphi= (\varphi_1, \dots, \varphi_p)$. 
After replacing $u = \varphi(x)$ into the differential equations $\dot{x}_i$ of $G$:
\begin{equation}\label{eq:dynG}
 \dot{x}_i = g_i(x)= f_i (x, \varphi(x)), \ i =1, \dots, n,
\end{equation}
we obtain a dynamical system associated to the network $G_{red}$ with mass-action kinetics, with reaction rate constants defined by the application 
\begin{equation}\label{eq:T}
T:
\R_{>0}^{\# \mathcal{R}_G} \to \R_{>0}^{\# \mathcal{R}_{G_{red}}}
\end{equation}
given by the following assignments:
\begin{equation}\label{eq:tau}
 T(\kappa)_{y y'} = \tau(\kappa)_{y y'} =\kappa_{yy'} + \sum_{j=1}^p \kappa_{U_j y'} \, \mu_{j,y}(\kappa),
\end{equation}
depending on the reaction rate constants $\kappa$ of $G$ \cite[Theorem 3.2]{fw13}. Here, $\kappa_{yy'}$ is positive when $y\overset{\kappa_{yy'}}{\longrightarrow} y'$ in $G$ (and $\kappa_{yy'}=0$ otherwise), and $\kappa_{U_j y'}$ is positive if $U_j\overset{\kappa_{U_j y'}}{\longrightarrow} y'$ in $G$ (and $\kappa_{U_j y'}=0$  otherwise), where $\mu_{j,y}$ is in \eqref{eq:concentracioninter}.

It is important to note that if the network $G_{red}$ has reaction rate constants $\tau(\kappa)$ as in \eqref{eq:tau}, the steady states of the mass-action chemical reaction systems defined by $G$ and $G_{red}$ are in one-to-one correspondence via the projection $\pi(x,u) = x$.

\begin{example}[Example~\ref{ex:int_cont}, continued]
Using~\eqref{eq:tau}, we can express the reaction rate constants $\tau$ in terms of the reaction rate constants $\kappa$.  We then have $\tau_1=\kappa_3\, \mu_{1,y_1}(\kappa)=\frac{\kappa_1\kappa_3}{\kappa_3+\kappa_5}$, $\tau_2=\kappa_7\, \mu_{3,y_2}(\kappa)=\frac{\kappa_4\kappa_5}{\kappa_3+\kappa_5}$, $\tau_3=\kappa_7\, \mu_{3,y_1}(\kappa)=\frac{\kappa_1\kappa_5+\kappa_2\kappa_5+\kappa_2\kappa_3
}{\kappa_3+\kappa_5}$. 
\end{example}

From the proofs of Theorems~3.1 and~3.2 in~\cite{fw13} it can be immediately inferred that the system $f_1, \dots, f_{n+p}$ can be transformed via {\em reversible} linear operations into the system defined by the equalities~\eqref{eq:concentracioninter} and $g_1, \dots, g_n$.  We state  
their result in the following lemma.

\begin{lemma}[\cite{fw13}]\label{lem:linear}
Let $G$ be a  chemical reaction network with set of intermediate species $U_1,\dots, U_p$, core species $X_1, \dots,X_n$ and rate constants $\kappa$ with $f_{\kappa, i}, i=1, \dots, n+p$ as in~\eqref{eq:CRN}, consider the reduced network $G_{red}$
with rate constants $\tau=T(\kappa)$ with $g_{\tau, i} = g_i, i=1, \dots, n$ and let $v_i = u_i -\varphi(x)$, $i = 1,\dots, p$, with $g_i$ and $\varphi$ as in~\eqref{eq:dynG}.
Then, the system
\[f_{\kappa,1}(x,u) = \dots = f_{\kappa, n+p}(x,u) =0\]
is equivalent via $\R$-linear operations to the system
\[g_{\tau,1}(x) = \dots = g_{\tau, n}(x) = v_1(x,u) = \dots= v_p(x,u) =0.\]
\end{lemma}

\subsection{Non-confluent networks and boundary steady states}\label{ssec:bss}

When studying the dynamics of a mass-action kinetics system in the positive orthant, it is important to understand the ocurrence of boundary steady states
in stoichiometric compatibility classes that intersect the positive orthant.
These are called {\em relevant} boundary steady states. We now relate the absence of relevant boundary steady states associated with a network with intermediate species and its corresponding reduced network.

\begin{definition}[Non-confluent networks]\label{def:nc}
 Let $G$ be a chemical reaction network with intermediate species $\mathcal{I}$. A  core complex $y$ with $y \uri U$ and $U \in \mathcal I$, is called an {\em input} of $U$. We say that $G$ is non-confluent 
 if any $U \in \mathcal I$ has  a unique input.
\end{definition}

Standard examples of non-confluent reaction networks include the distributive sequential phosphorylations (see Figure~\ref{fig:nsite}). 
More in general, if we add intermediates to a network $G'$ with one of the following local shapes around any intermediate $U$, with $y, y' \in \mathcal C_{G'}$, we get a non-confluent extension:
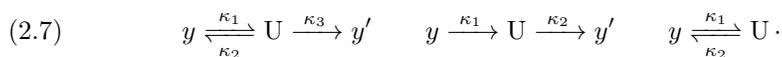
\begin{equation}\label{eq:ncexamples}
y\ce{
<=>[\kappa_1][\kappa_2]
U
->[\kappa_3]
} y' \qquad y\ce{
->[\kappa_1]
U
->[\kappa_2]
} y' \qquad y\ce{
<=>[\kappa_1][\kappa_2]
U.
}
\end{equation}

\begin{remark}\label{rem:uk}
Given a non-confluent reaction network $G$ and $U_k$ an intermediate species,  let $y^{(k)}$ be the unique complex in $G$ with $y^{(k)} \uri U_k$. Then, by
~\eqref{eq:concentracioninter} we have that there exist a positive $\mu_k$ such that 
\begin{equation} \label{eq:ukbinom}
 u_k=\mu_kx^{y^{(k)}}.
\end{equation}
\end{remark}

\smallskip

\begin{proposition}\label{prop:bss} 
Let $G$ be a non-confluent chemical reaction network with intermediate species $\mathcal{I}$ and rate constants $\kappa$. Consider the reduced network $G_{red, \mathcal{I}}$ with rate constants $T(\kappa)$ as in~\eqref{eq:tau}. 
The reduced network $G_{red,\mathcal{I}}$  has no relevant boundary steady states if and only if $G$ does not have relevant boundary steady states.
\end{proposition}

\begin{proof}
Denote $\Sp_G=\{X_1,\dots, X_n,U_1,\dots,U_p\}$, with intermediates $\mathcal{I}=\{U_1,\dots,U_p\}$.
Consider conservation laws $\ell$ of $G_{red,\mathcal{I}}$ and $\bar{\ell}$ of  $G$ as described in~\eqref{eq:cons_laws_G}.
Let $(x,u)\in\R_{\geq 0}^{n+p}$ be a relevant boundary steady state of $G$. Then, there exists a positive point $(x^0,u^0) \in \R_{>0}^{n +p}$ such that 
\begin{equation}\label{eq:comparing_cons_laws}
 \bar{\ell}(x,u)=\bar{\ell}(x^0,u^0).
\end{equation}
We can assume that there is an index $j$ such that $x_j=0$.  Indeed, if some $u_i=0$, we have by~\eqref{eq:concentracioninter} that some coordinate of $x$ must vanish because all coefficients $\mu_{i,y}$ are nonnegative and at least one of them is strictly positive by the definition of intermediate species. Moreover, if $(x,u) \not=0$, we cannot have that all $x_i=0$ because this implies by the same relations that all $u_k=0$. Then, $\pi(x,u)=x$ is a nonzero boundary steady state of $G_{red,\mathcal{I}}$.  

We claim that $x$ is a relevant boundary steady state of $G_{red,\mathcal{I}}$.  We need to find a positive point $(x')^0 \in \R_{>0}^n$ such that
\begin{equation}\label{eq:primas}
\ell(x) = \ell((x')^0).
\end{equation}
Given $k$, let $y^{(k)}$ be the unique core complex with $y^{(k)} \uri U_k$. Then, $u_k=\mu_kx^{y^{(k)}}$ for some positive constant $\mu_k$ by~\eqref{rem:uk}.
Now, from~\eqref{eq:cons_laws_G} and~\eqref{eq:comparing_cons_laws} we can write 
\[ \ell(x)=\ell(x^0+\sum_{k=1}^p (u^0_k-\mu_k x^{y^{(k)}}) \,  y^{(k)}) = \ell(v^0),\]
with $v^0=x^0+\sum_{k=1}^p (u^0_k-\mu_k x^{y^{(k)}}) y^{(k)}$. Consider the affine combination $v_\alpha=\alpha v^0+(1-\alpha) x$. 
As $x^{y^{(k)}}\neq 0$ if and only if $\{\ell: y^{(k)}_\ell\neq 0\} \subseteq  \{\ell: x_{\ell}>0\}$, whenever $v^0_\ell<0$ necessarily $x_\ell>0$. 
We can then take $\alpha_0$ small enough such that $v_{\alpha_0} >0$ and pick the positive vector $(x')^0=v_{\alpha_0}$.

Now, let $x\in\R_{\geq 0}^n$ be a relevant boundary steady state of $G_{red,\mathcal{I}}$.  This means that there exists a positive point $x^0 \in \R_{>0}^n$ such that $\ell(x^0) = \ell(x)$ for any conservation law $\ell$ of $G_{red,\mathcal{I}}$. 
We can complete $x$ to a steady state $(x,u)\in\R_{\geq 0}^{n+p}$ by~\eqref{eq:concentracioninter} and we are then looking for a positive vector $((x^0)',(u^0)')\in \R_{>0}^{n +p}$ such that $\bar{\ell}((x^0)',(u^0)')=\bar{\ell}(x,u)$ for any conservation law $\ell$ of $G_{red,\mathcal{I}}$. We can pick a positive vector $(\alpha_1,\dots,\alpha_p)$ with each coordinate $\alpha_k$ small enough such that $x^0-\sum_{k=1}^p\alpha_ky^{(k)}$ is a positive vector, where $y^{(k)}$ is a core complex  in the same connected component as $U_k$. Then $\bar{\ell}(x,u)=\bar{\ell}(x^0,u)=\ell(x^0+\sum_{k=1}^p u_k y^{(k)})$, which equals $\ell\left(x^0-\sum_{k=1}^p \alpha_ky^{(k)} +\sum_{k=1}^p (u_k +\alpha_k) y^{(k)} \right)$. Consider then $(x^0)'=x^0-\sum_{k=1}^p \alpha_ky^{(k)}$ and $(u^0)'$ such that $(u^0)_k'=u_k +\alpha_k$.
\end{proof}

\subsection{Lifting steady states} \label{ssec:lifting}

We now study the lifting of steady states from a reduced network $G_{red}$ to an extended network $G$. This relation was first studied in Theorem~5.1 in~\cite{fw13}.
We also refer the reader to~\cite{Murad23,Murad22,Murad18}.

We need to introduce the notion of non-degenerate steady states when the stoichiometric subspace is not the whole space. 
In this case, there are non-trivial linear relations between $f_1,\dots, f_s$ and the usual condition of non-degeneracy of a steady state $x^*$ given by the non-vanishing of the Jacobian determinant $\det J(f)(x^*)$  needs to be extended.

\begin{definition}\label{def:nondeg}
Consider a dynamical system as in~\eqref{eq:CRN} with associated stoichiometric subspace $\mathcal S$.  A steady state $x^* \in \R^s_{>0}$ is said to be non-degenerate if $\ker(J(f))(x^*) \cap \mathcal S = \{0\}$. 
\end{definition}

Standard linear algebra arguments show the following result.

\begin{lemma} \label{lem:nondeg}
Given a dynamical system as in~\eqref{eq:CRN} with associated stoichiometric subspace $\mathcal S = \{x \in \R^s \, : \, W x=0\} $ of dimension $s-d$, and a positive steady state $x^*$, the following assertions are equivalent:
\begin{itemize}
 \item[(i)] $x^*$ is non-degenerate.
 \item[(ii)] \label{(2)} There exist $s-d$ linearly independent functions $f_{i_1}, \dots, f_{i_{s-d}}$ such that  the  $s \times s$ 
 matrix with first $s-d$ rows given by the Jacobian of these  functions evaluated at $x^*$ and the last $d$ rows corresponding  to the matrix $W$,  has nonzero determinant. 
\end{itemize}
\end{lemma}
Of course, if condition~(ii) holds, the determinant constructed as above for {\em any} choice of $s-d$ linearly independent functions $f_i$ will also be nonzero.

\medskip

Let $G$ be a chemical reaction network with reaction rate constants $\kappa$ and intermediate species $\mathcal{I}=\{U_1,U_2,\dots,U_p\}$. 
Consider the reduced network $G_{red} = G_{red, \mathcal{I}}$ with rate constants defined by the application $T$ in~\eqref{eq:T}.
In Theorem~5.1 in~\cite{fw13} it is shown that if  $G$ has reaction rate constants $\kappa^0$ and we  consider the reduced network $G_{red}$ with reaction rate constants $\tau^0 =T(\kappa^0)$, then if $G_{red}$ has $m$ non-degenerate steady states in a stoichiometric compatibility class, then there is a curve of rate constants $\kappa'$ with $T(\kappa')=\tau^0$ such that the associated system has at least $m$ non-degenerate steady states in a compatibility class of $G$.

We extend their result in Theorem~\ref{thm:IFT}, by describing regions in the space of parameters for which we can lift the steady states of the reduced network to the original network. Given $\tau^0$ in the image of $T$, we denote the fiber of $\tau^0$ by
\begin{equation} \label{eq:Ftau}
{\mathcal F}_{\tau^0}=\{\kappa>0 \, : \, T(\kappa)=\tau^0 \}.
\end{equation}
In the next theorems, we describe open sets in ${\mathcal F}_{\tau^0}$ such that multistationarity is lifted to the mass-action systems of all the extended networks with rate constants lying in these open sets.

\begin{theorem}\label{thm:IFT} 
Let $G$ be a chemical reaction network with intermediate species $\{U_1,U_2,\dots,U_p\}$ with associated map $T$ as in~\eqref{eq:T}. Consider the reduced network $G_{red}$ with rate constants $\tau^0 \in {\rm im}(T)$.

Given $\varepsilon >0$ and $\mu_{i,y}$ as in~\eqref{eq:concentracioninter}, the open set of the fiber ${\mathcal F}_{\tau^0}$:
\[
{\mathcal F}_{\tau^0,\varepsilon} =\{\kappa \in  {\mathcal F}_{\tau^0}\, : \, \mu_{i,y}(\kappa) 
< \varepsilon\ \, \, \forall y\in \mathscr{C}_{G_{red}},\, i=1,\dots,p\}\]
is nonempty.

Moreover, fix $c_1,\dots,c_d\in\R$ and consider the stoichiometric compatibility class $\mathcal S_c$ defined by the equations $\ell_1(x)=c_1,\dots,\ell_d(x)=c_d$, where $\ell_1(x), \dots, \ell_d(x)$ is a basis of conservations laws of the system associated with $G_{red}$.

 \smallskip
 
\noindent \begin{minipage}{0.6\textwidth}
 Then, if  $G_{red}$ has $m$ non-degenerate positive steady states in $\mathcal S_c$, there exists a positive value $\varepsilon_0$ such that for all $\kappa\in \mathcal F_{\tau^0,\varepsilon}$ with $0<\varepsilon<\varepsilon_0$, there are at least $m$ non-degenerate positive steady states of $G$ with reaction rate constants $\kappa$ in the stoichiometric class of the system associated with $G$ defined by $\bar{\ell}_1(x,u)=c_1, \dots,\bar{\ell}_d(x,u)=c_d$, where $\bar{\ell}_1, \dots, \bar{\ell}_d$ is a basis of conservation laws of the system associated with $G$ obtained from $\ell_1(x), \dots, \ell_d(x)$ as in~\eqref{eq:cons_laws_G}.
\end{minipage}
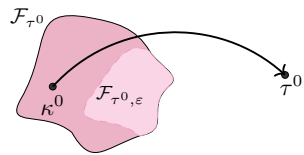
\begin{minipage}{0.35\textwidth}
 \begin{center}
 \tikzset{
 flecha/.style={
              ->,
              thick},
              }
 \scalebox{0.9}{
    \begin{tikzpicture}
        \node (O) at (-5,-2) {};
		\coordinate (A) at ($(O)+(1.75,1.3)$);
		\coordinate (B) at ($(O)+(2.75,1.25)$);
		\coordinate (C) at ($(O)+(3.2,2)$);
		\coordinate (D) at ($(O)+(2.9,2.5)$);
		\coordinate (E) at ($(O)+(2,3)$);
		\coordinate (F) at ($(O)+(1,1.8)$);
		\node (G) at ($(O)+(1.5,1.85)$) {};
		\node (H) at ($(O)+(3,1.5)$) {};
        \draw[fill=purple!30] (A) to[out=-60, in=200, looseness=1.5] (B) to[out=10, in=220, looseness=1.5] (C) to[out=40, in=-60, looseness=1] (D) to[out=120, in=-10, looseness=1] (E) to[out=170, in=60, looseness=1.5] (F) to[out=240, in=140, looseness=1.5] cycle;
        \draw[color=magenta!50,fill=magenta!20,dashed]  (B) to[out=10, in=220, looseness=1.5] (C) to[out=40, in=-60, looseness=1] (D) to[out=150, in=50, looseness=1]  ($(F)+(1,0)$) to[out=240, in=140, looseness=1.5] cycle;
        \node[] (t) {$\tau^0$};
        \node[inner sep=8pt,left=2cm of t] (k1) {};
        \node[inner sep=1pt, above left= 0.05mm of H] (k) {\footnotesize $\mathcal{F}_{\tau^0,\varepsilon}$};
        \draw[fill=black!80] ($(k)+(-1,0.15)$) circle (0.5mm) node[below] {$\kappa^0$};
        \draw[fill=black!80] ($(t)+(-0.1,0.15)$) circle (0.5mm);
        \path[flecha] ($(k)+(-1,0.15)$) edge [bend left=45] ($(t)+(-0.1,0.15)$);
        \draw node[left=5mm of E]{\small $\mathcal{F}_{\tau^0}$};
   \end{tikzpicture}
  }
 \end{center}
 \end{minipage}

\end{theorem}

\begin{proof} 
The proof is an extension of that in Theorem~5.1 in~\cite{fw13}. 
The species of $G$ are ordered as: $X_1,\dots, X_n, U_1,\dots, U_p$, and the species of $G_{red}$ are $X_1,\dots, X_n$. 
Given reaction rate constants $\kappa=(\kappa_{yy'})$ and $\theta\in\R_{>0}$, they define a curve of rate constants $\kappa^{\theta}=(\kappa^{\theta}_{yy'})$ by $\kappa^{\theta}_{yy'}= \frac{\kappa_{yy'}}{\theta}$ if $y$ is an intermediate species and $\kappa^{\theta}_{yy'}=\kappa_{yy'}$ otherwise. Then it can be checked that for any $\kappa \in \mathcal F_{\tau^0}$,  
$\kappa^{\theta}\in \mathcal{F}_{\tau^0}$ and
$\mu^{\theta}_{i,y}=\theta \,  \mu_{i,y}$.
Therefore, given any positive $\varepsilon$, if we take $\theta$ small enough it happens that $\kappa^{\theta}\in \mathcal{F}_{\tau^0, \varepsilon}$, and then, $\mathcal{F}_{\tau^0,\varepsilon}$ is a nonempty open set of $\mathcal{F}_{\tau^0}$.

\smallskip

Let now $g_1, \dots, g_n$ be the polynomials defined in~\eqref{eq:dynG} associated to the network $G_{red}$. As there exists a non-degenerate positive steady state of $G_{red}$ in the stoichiometric compatibility class $\mathcal S_c$, without loss of generality we can assume by Lemma~\ref{lem:nondeg} that $g_{d+1}, \dots, g_n$ are linearly independent and generate the $\R$-vector space generated by all the $g_i$'s. It follows that $f_{d+1}, \dots, f_n$ are also linearly independent and generate all the $f_i$'s. As we want to prove that there exist $\kappa$ such that the system $\bar{\ell}_1(x,u) - c_1 =\dots=\bar{\ell}_d(x,u) - c_d=
f_{\kappa,1}(x,u)=\dots=f_{\kappa,n+p}(x,u)=0$ has $m$ positive non-degenerate solutions, we can consider the solutions of the equivalent system defined by
\begin{small}
\begin{equation}\label{sysaugmented:G} 
f_{c,\kappa}(x,u)=(\bar{\ell}_1(x,u) - c_1,\dots,\bar{\ell}_d(x,u) - c_d, 
f_{\kappa,d+1}(x,u),\dots,f_{\kappa,n+p}(x,u)),
\end{equation}
\end{small}%
where $f_{\kappa,i}(x,u)$ is the $i$-th coordinate of the function $f_{\kappa}(x,u)$ as in~\eqref{eq:CRN}. 
Then, a vector $(x,u)$ is a steady state of $G$ for the reaction rate constants $\kappa$ and for the stoichiometric compatibility class defined by $c_1,\dots,c_d$, if and only if $f_{c,\kappa}(x,u)=0$.

\smallskip

Analogously, we consider the system augmented by conservation laws corresponding to the network $G_{red}$. A vector $x$ is a steady state of $G_{red}$ for reaction rate constants $\tau$ in the stoichiometric compatibility class defined by $c_1,\dots,c_d$ if and only if $x$ is a zero of the following function:
\begin{small}
\begin{equation}\label{sysaugmented:G'} 
g_{c,\tau}(x)= (\ell_1(x) - c_1,\dots,\ell_d(x) - c_d, g_{\tau,d+1}(x),\dots,g_{\tau,n}(x)),
\end{equation}
\end{small}%
where 
\begin{small}
\[g_\tau(x) = \sum_{y\rightarrow y'\in \mathcal{R}_{G_{red}}} \tau_{yy'}x^y(y'- y).\]
\end{small}%

In the proof of Theorem 5.1 in~\cite{fw13} and recalled in Lemma~\ref{lem:linear}, it is shown that the system in~\eqref{sysaugmented:G} can be transformed via linear operations into the equations below.  That is, there exists an invertible matrix $K^*$ such that:
\begin{small}
\begin{align}\label{sysaugmented:G_subs} 
\nonumber f_{c,\kappa}(x,u)=& K^* \, (\bar{\ell}_1(x,\varphi(x)) - c_1,\dots,\bar{\ell}_d(x,\varphi(x)) - c_d,\\
\nonumber & f_{\kappa,d+1}(x,\varphi(x)),\dots,f_{\kappa,n}(x,\varphi(x)), u_1-\varphi_1(x),\dots,u_p-\varphi_p(x))^t\\
=& K^* \, (\ell_1(x)-c_1 + \sum_{k=1}^p a_{1k} \, \varphi_k(x),\dots,\ell_d(x)-c_d + \sum_{k=1}^p a_{dk} \, \varphi_k(x),\\
\nonumber & f_{\kappa,d+1}(x,\varphi(x)),\dots,f_{\kappa,n}(x,\varphi(x)),u_1-\varphi_1(x),\dots,u_p-\varphi_p(x))^t,
\end{align}
\end{small}%
where $\varphi_i(x) =\sum_{y\in{\mathscr{C}_{G_{red}}}}{\mu_{i,y}(\kappa) x^y}$ for $i=1, \dots, p$,  $\varphi= (\varphi_1, \dots, \varphi_p)$, as in~\eqref{eq:dynG}, and $a_{ik}=\ell_i(y^{(k)})$ with $y^{(k)}$ a core complex in the same connected component as $U_k$. Then, the Jacobian matrix of ${f}_{c,\kappa}$ evaluated at $(x,u)$ is nonsingular if and only if the Jacobian matrix of the right-hand side equations evaluated at $(x,u)$ is nonsingular.

\smallskip

Note also that, if $T(\kappa)=\tau$, then $f_{\kappa,d+1}(x,\varphi(x))=g_{\tau,d+1}(x),\dots,f_{\kappa,n}(x,\varphi(x))=g_{\tau,n}(x)$.
Take now all the nonzero coefficients $\mu_{i,y}$, for all $i=1,\dots,p$ and all complexes $y\in G_{red}$, and let $N$ be the number of nonzero $\mu_{i,y}$ coefficients. Let $\mu\in \R_{>0}^N$ be the vector with coordinates $\mu_{i,y}$, in some order. Fix $c_1,\dots,c_d$, $\tau^0 \in {\rm im}(T)$ and consider the function $F_{c,\tau^0}\colon\R^n\times \R^p \times \R^N  \to \R^{n+p}$ defined by:
\begin{small}
\begin{equation}\label{sysreplace:G}
F_{c,\tau^0,i}(x,u,\mu)= \left\{ \begin{array}{lc}
            \ell_i(x)-c_i + \sum_{k=1}^p\sum_{y\in{\mathscr{C}_{G_{red}}}}{\mu_{k,y} a_{ik}x^y}, \ i=1,\dots,d, \\
            g_{\tau^0,i}(x), \ i=d+1,\dots,n,\\ 
            u_{i-n} - \sum_{y\in{\mathscr{C}_{G_{red}}}}{\mu_{{i-n},y} x^y},\ i=n+1,\dots,n+p.
             \end{array}
   \right.
\end{equation}
\end{small}%

Assume now that $G_{red}$ has $m$ nondegenerate positive steady states $x^{(i)}\in \R^n_{>0}$, $i=1,\dots,m$ in the stoichiometric compatibility class defined by the total amounts $c_1,\dots,c_d$, and for the reaction rate constants $\tau^0$.

\smallskip

For $\mu=0$, $F_{c,\tau^0}(x^{(i)},0,0)=0$, because $\ell_j(x^{(i)})=c_j$ for $j=1,\dots,d$ and $g_{\tau^0,j}(x^{(i)})=0$ for all $j=d+1,\dots,m$, and the Jacobian matrix of $F_{c,\tau^0}(x,u,0)$ has the form:
\[J_{(x,u)}(F_{c,\tau^0})(x,u,0)=\begin{pmatrix}
J_x(g_{c,\tau^0}) & 0 \\
0 & I_p \\
\end{pmatrix}.\]

Since the steady states $x^{(i)}$ are nondegenerate, $J_x(g_{c,\tau^0})$ evaluated at $x^{(i)}$ is nonsingular for each $i=1,\dots,m$. 
Then, $J_{(x,u)}(F_{c,\tau^0})$ evaluated at $(x^{(i)},0,0)$ is nonsingular for each $i=1,\dots,m$. By the Implicit Function Theorem applied to $F$ at $(x^{(i)},0,0)$, there exists an open set $\mathcal{U}_i\subset \R^N$, with $0 \in \mathcal{U}_i$, and an open set $\mathcal{V}_i\in \R^n\times \R^p$, with $(x^{(i)},0)\in \mathcal{V}_i$ such that for all $\mu \in \mathcal{U}_i$, there is a vector $(x^{(i)}(\mu),u^{(i)}(\mu))\in \mathcal{V}_i$ such that $F_{c,\tau^0}(x^{(i)}(\mu),u^{(i)}(\mu),\mu)=0$. 

\medskip

Because $x^{(i)}>0$ and $x^{(i)}$ is a nondegenerate steady state, we can take the open set $\mathcal{U}_i$ such that $x^{(i)}(\mu)>0$ and $J_{(x,u)}(F_{c,\tau^0})(x^{(i)}(\mu),u^{(i)}(\mu),\mu)$ is nonsingular for all $\mu \in \mathcal{U}_i$. We take $\mathcal{U}_i^+=\mathcal{U}_i\cap \R^N_{> 0}$. Since $x^{(i)}(\mu)>0$, it follows that $u^{(i)}(\mu)>0$ for all $\mu \in \mathcal{U}_i^+$, by construction. 
Because the $x^{(i)}$ are distinct, we can moreover choose the open sets $\mathcal{U}_i$ (smaller if needed, contained in the  original $\mathcal{U}_i$) such that $\cap_{i=1}^m \mathcal{V}_i=\emptyset$.

\medskip

Now we take $\mathcal{U}=\cap_{i=1}^m \mathcal{U}_i^+$. If $\kappa>0$ is such that $T(\kappa)=\tau^0$, and $\mu=\mu(\kappa)$ as in~\eqref{eq:tau} is such that $\mu \in \mathcal{U}$, then the original network $G$ has $m$ nondegenerate positive steady states  $(x^{(i)}(\mu),u^{(i)}(\mu))$, $i=1,\dots,m$ in the stoichiometric compatibility class defined by $c_1,\dots,c_d$.

\end{proof}

\begin{theorem}\label{thm:IFT2} 

With the same hypotheses as in Theorem~\ref{thm:IFT}, assume moreover that $G$ is an extension network obtained from $G_{red}$ by adding reactions of any of the forms depicted in~\eqref{eq:ncexamples}. 
Consider the reduced network $G_{red}$ with rate constants $\tau^0 \in {\rm im}(T)$. Given $M >0$, the open set of the fiber ${\mathcal F}_{\tau^0}$:
\[
{\mathcal F}_{\tau^0,M} =\{\kappa \in  {\mathcal F}_{\tau^0}\, : \, \kappa_{U_iy'}>M\ \, \text{for all reactions } U_i\to y',\, i=1,\dots,p\}\]
is nonempty.

Moreover, if $G_{red}$ has $m$ non-degenerate positive steady states in the stoichiometric compatibility class defined by the equations $\ell_1(x)=c_1,\dots,\ell_d(x)=c_d$, there exists a positive value $M_0$ such that for all $M \ge M_0$ there are at least $m$ non-degenerate positive steady states of $G$ with reaction rate constants $\kappa \in {\mathcal F}_{\tau^0,M}$ in the stoichiometric class of the system associated with $G$  defined by $\bar{\ell}_1(x,u)=c_1, \dots,\bar{\ell}_d(x,u)=c_d$, where $\bar{\ell}_1, \dots, \bar{\ell}_d$ are obtained from $\ell_1(x), \dots, \ell_d(x)$ as in~\eqref{eq:cons_laws_G}.

\end{theorem}

\begin{proof}
  If $G$ is an extension network obtained by adding reactions of any of the forms depicted in~\eqref{eq:ncexamples}, by Remark~\ref{rem:uk} we know that $u_i=\mu_ix^{y^i}$, where in the first two cases we have:
 \begin{center}
 \begin{tabular}{|c|c|c|}
 \hline
 Extension & $\mu_i=\mu_i(\kappa)$ & $\tau_{yy'}=T_{yy'}(\kappa)$\\
 \hline
 $y\ce{
<=>[\kappa_{yU_i}][\kappa_{U_iy}]
U_i
->[\kappa_{U_iy'}]
} y'$ &  
$\frac{\kappa_{yU_i}}{\kappa_{U_iy}+\kappa_{U_iy'}}$ &
$\kappa_{U_iy'}\mu_i$ \\
\hline
$y\ce{
->[\kappa_{yU_i}]
U_i
->[\kappa_{U_iy'}]
} y' $ & 
$\frac{\kappa_{yU_i}}{\kappa_{U_iy'}}$ &
$\kappa_{U_iy'}\, \mu_i$ \\
\hline
\end{tabular}
\end{center}

In the first type of extension, given any choice of positive
rate constants $\kappa_{U_iy}, \kappa_{U_iy'}$ if we pick $\kappa_{yU_i}=\tau_{yy'}\frac{\kappa_{U_iy}+\kappa_{U_iy'}}{\kappa_{U_iy'}}$ then $T_{yy'}(\kappa) = \tau_{yy'}$. For the second type of extension we can also solve for $\kappa_{yU_i}$ for any choice of $\kappa_{U_iy'}$. In the case of a canonical extension, $y\ce{<=>[\kappa_{yU_i}][\kappa_{U_iy}]} U_i$ we have that ${\mathcal F}_{\tau^0}$ is the whole set of parameters $\kappa$ because $y =y'$.

Consider $\varepsilon_0$ as in the statement of Theorem~\ref{thm:IFT}. 
It is clear that there exists $M_0 >0$ such that if $\kappa_{U_iy'}>M_0$ then $\mu_i<\varepsilon_0$ for every $i$, and the result follows by Theorem~\ref{thm:IFT}.
\end{proof}

\begin{remark}\label{rem:unasola}
 Note that the proof of Theorem~\ref{thm:IFT2} shows that for extensions of the form
 \[y\ce{
<=>[\kappa_{yU_i}][\kappa_{U_iy}]
U_i
->[\kappa_{U_iy'}]
} y',\]
it is enough that the rate constants $\kappa_{U_iy'}$ are sufficiently big to imply multistationarity. That is, multistationarity can be lifted when the {\em catalytic} reaction is ``sufficiently fast''.
\end{remark}

\section{Circuits of multistationarity}\label{sec:circuits}

In this section, we retrieve and simplify the results in~\cite{SFeliu}  about  {\em complete binomial networks} (see \cite[Definition~2.6]{SFeliu}).    
The condition of having a complete binomial network is a very general setting for the theoretical results in~\cite{SFeliu}, but it cannot be checked in general.
We consider a subclass of complete binomial networks in the sense of Definition~2.6 in~\cite{SFeliu} that we call {\em linearly equivalent to a binomial network} (see Definition~\ref{def:lebn} below) for which the technical hypotheses of the main results can be ensured (see Propositon~\ref{prop:CK}).   These conditions are more restrictive but we will show that they are satisfied by plenty of interesting biochemical networks, including the examples in all the following sections and all the examples in~\cite{SFeliu}. 

\medskip

As we are interested in the existence of positive steady states,  we first introduce a basic necessary condition. We write system~\eqref{eq:CRN} in the following form:
\begin{equation}\label{eq:CRNN}
\dot{x} \, = \, N \cdot R_\kappa(x),
\end{equation}
where we order the set of reactions $\mathcal R_G$, $R_\kappa(x)$ is the vector of size equal to the cardinality $r$ of $\mathcal R_G$ defined as follows: if the $i$-th reaction is $y \to y'$ then $R_\kappa(x)_i = \kappa_{y y'} x^y$ and $N$ is the integer matrix whose $i$-th column equals $y'-y$, known as the {\em stoichiometric matrix} of the network.

\begin{definition}\label{def:consistent}
A network is said to be {\em consistent} when $\ker(N) \cap \R^r_{>0} \neq \emptyset$.
\end{definition}
Note that the existence of a positive steady state $x$ of a reaction network with positive rate constants $\kappa$, gives the positive vector $R_\kappa(x)$ in $\ker(N)$ and so any network with a positive steady state must be consistent.
An interesting consequence of consistency is the following well-known result (see e.g. \cite{CFR}).

 \begin{lemma}\label{lem:clarke}
Assume that a reaction network with stoichiometric matrix $N$ is consistent and let $v \in \ker(N)\cap \R^r_{> 0}$. For any choice of $x \in \R^s_{>0}$ there exists a choice of rate constants $\kappa$ such that  $f_\kappa(x) =0$, that is, for which $x$ is a steady state of the associated system.
 \end{lemma}

\begin{proof}
We need to find a value of $\kappa \in \R^r_{>0}$ such that $f_\kappa (x) = N \cdot R_\kappa(x)$.
It is then enough to find a value of $\kappa$ for which $R_\kappa(x) = v$, which is satisfied by  setting $\kappa_{yy'}= v_{yy'}/x^y$ for every reaction $y \to y'$. 
\end{proof}

\begin{lemma} \label{lem:cons}
  Given any consistent reaction network $G'$, any extension network $G$ obtained by adding reactions of any of the forms depicted in~\eqref{eq:ncexamples} is also consistent. 
  \end{lemma}

\begin{proof}
It is enough to show that given a consistent reaction network $G'$, the extension network $G$ obtained by adding a new complex $U$ to any choice of reaction $y \to y'$ in $G'$ and intermediate reactions of any of the forms depicted in~\eqref{eq:ncexamples} is also consistent.
 Let $N'$ be the stoichiometric matrix of $G'$. As $G'$ is consistent there exists a vector $v' \in \ker(N')\cap \R^{r'}_{> 0}$. We show how to get
 from $v'$ a positive vector in the kernel of the stoichiometric matrices of these  extensions.
 
 If $G$ is obtained by adding a (reversible) reaction of the form $y\ce{<=> U }$, there are two opposite columns in the stoichiometric matrix $N$ of $G$ for these new reactions. Then, the vector 
 $v\in\R^r_{>0}$ which coincides with $v'$ in the entries that correspond to reactions in $G'$ and $1$'s in the entries corresponding to the added columns belongs to $\ker(N)\cap \R^{r}_{> 0}$.
 
 Assume we replace the reaction $y \to y'$  in $G'$ (which we assume to be the $r'$-th reaction). 
 Assume first that $G$ is obtained by replacing this reaction  by the two reactions $y\to U\to y'$. Note that slightly abusing the notation, we have that  $y'-y = (y'- U) + (U - y)$. Then, the vector $v\in\R^r_{>0}$ which coincides with $v'$ in the entries that correspond to reactions in $G'$ except for the reaction $y \to y'$ and $v'_{r'}$ repeated in the entries corresponding to the two added columns, belongs to $\ker(N)\cap \R^{r}_{> 0}$.
 If instead $G$ is obtained by deleting the reaction $y \to y'$ and adding the reactions $y\ce{<=> U}\to y'$, then (again with a slight abuse of notation) ${y'}-{y}=({y'} - U)+2(U-{y})+({y}-U)$.  Then, the vector $v\in\R^r_{>0}$ which coincides with $v'$ in the entries that correspond to reactions in $G'$ except for the reaction $y \to y'$, and $v'_{r'}$ repeated in the coordinates that correspond to $U\to y$ and $U\to y'$ and $2v'_{r'}$ in the coordinate that corresponds to $y\to U$, lies in $\ker(N)\cap \R^{r}_{> 0}$.
 
\end{proof}

 We next  introduce the definition of extended systems.

\begin{definition}\label{def:extsys}
 Let $G'$ be a core network definining a mass-action system with vector of rate constants $\tau$. Given a set of intermediates $\mathcal I$, we say that the extended network $G= G'_{ext, \mathcal I}$ with rate constants $\kappa$ is an extended system if $\tau= \tau(\kappa)$ are related by the equations~\eqref{eq:tau}.
\end{definition}

The following lemma is a consequence of item (ii) in~\cite[Prop. 5.3]{SFeliu}.

\begin{lemma} \label{lem:ext}
 Given a core network $G'$ with rate constants $\tau$ and any extended network $G= G'_{ext, \mathcal I}$ obtained by adding reactions of any of the forms depicted in~\eqref{eq:ncexamples}, there exist (infinitely many) choices of vectors of rate constants $\kappa$ which define an extended system.
\end{lemma}

The main question we address is: given a core network, which are the minimal sets of intermediates required for the associated extended network to have the capacity for multstationarity (see Definition~\ref{def:circuits}).
 From Proposition~\ref{prop:nc=canonical}, based on~\cite[Thm.~5.1]{fw13} and~\cite[Prop. 5.3 (ii)]{SFeliu}, we deduce that when $G'$ has two or more stoichiometrically compatible non-degenerate positive steady states, all the extensions by intermediates that arise from~\eqref{eq:ncexamples} are multistationary and the only minimal set is the empty set.

\medskip

We now define canonical extensions of networks  following~\cite{fw13}.

\begin{definition} \label{def:can}
Given a network $G'$, a subset $C \subset \mathcal C_{G'}$ and a choice of an intermediate complex  $U(y)$ for each $y \in C$, the associated canonical extension $G = G'_{ext,C} =G'_{ext, \{U(y) : y \in C\}}$ of $G'$ is obtained by adding the reactions $y\ce{<=>U(y) }$ for all $y \in C$.
\par
Given an extension $G$ of a network $G'$ , let $C \subset \mathcal C_G$ be the set of input complexes of $G$. We say that $G'_{ext,C}$ is the canonical extension of $G'$ associated to $C$.
\end{definition}
\noindent Note that canonical extensions are non-confluent (see \eqref{eq:ncexamples}).

\medskip
The following definition is based on~\cite[Definition~4.6]{SFeliu}.

\begin{definition}\label{def:circuits}
Given a network $G'$ and a subset $C' \subset \mathcal C_{G'}$, the {\em circuits of multistationarity associated to $G'$ and $C'$} are the minimal subsets $C$ of $C'$ (with respect to inclusion) for which the canonical extension $G'_{ext,C}$ is multistationary. In case $C' =\mathcal C_{G'}$, we say that they are the {\em circuits of multstationarity of $G'$}.
\end{definition}

In fact, the canonical extension associated to a subset of complexes is multistationary if and only if the other two non-confluent extensions
by intermediates in~\eqref{eq:ncexamples} are so.

\begin{proposition} \label{prop:nc=canonical}
Given a network $G'$ and a subset $C \subset \mathcal C_{G'}$, the canonical extension $G'_{ext,C}$ is multistationary if and only if adding to any reaction
$y  \to y'$ with $y \in C$ an intermediate $U$ with reactions either of the form $ y \to U \to y'$,  $ y \ce{<=> U } \to y'$ or $y \ce{<=> U}$ yields a multistationary network. 
 \end{proposition}

\begin{proof}
 The proof of the if direction follows from~\cite[Thm.~5.1]{fw13}  
 while the only if direction is a consequence of~\cite[Prop. 5.3 (ii)]{SFeliu}.
\end{proof}

\begin{remark} \label{rem:JS}
Note that by iterating the {\em allowed} addition of intermediates in
Proposition~\ref{prop:nc=canonical}, it is possible to extend multistationarity if we add several intermediate complexes. For instance,
we can iteratively add two intermediates $U_1, U_2$ to
$y \to y'$ by first adding $y \ce{<=>U_1} \to y'$ and then
$y \ce{<=>U_1} \to U_2 \to y'$. So, if the original network is multistationary, this last one is so too. Moreover,  if we add a backward reaction from $y'$ to $U_2$:
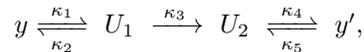
\begin{center}
 $y$\ce{
<=>[\kappa_1][\kappa_2]}
$U_1$
\ce{->[\kappa_3]}
$U_2$
\ce{<=>[\kappa_4][\kappa_5]
} $y'$,
\end{center}
we don't change the stoichiometric subspace. Then, the hypothesis of being non-confluent is not satisfied but Theorem~3.1 in~\cite{Joshi:Shiu:Atoms} asserts that this last network is also multistationary.

This can be applied for instance to lift multistationarity to {\em weakly irreversible} phosphorylation mechanisms from {\em strongly irreversible} mechanisms as defined in the article~\cite{Gu:irrev}, according to whether the product does rebind to the enzyme or not. We explain this concept via the following relevant mechanism of double sequential phosphorylation in~\cite{MHK}. The chemical species are the unphosphorylated substrate $S_0$, the singly phosphorylated substrate $S_1$, the doubly phosphorylated substrate $S_2$, two enzymes (the kinase $E$ and the phosphatase $F$) and $6$ intermediate species.

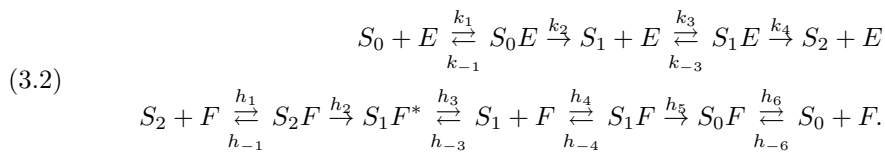
\begin{equation}\label{eq:red} 
\begin{split}
    {S_0} + {E}
    \underset{k_{-1}}{\overset{k_1}{\rightleftarrows}}
    {S_0E} 
    \overset{k_2}{\rightarrow}
    {S_1} + {E}
     \underset{k_{-3}}{\overset{k_3}{\rightleftarrows}}
    {S_1E}
    \overset{k_4}{\rightarrow}
    {S_2} + {E} \\
    {S_2} + {F}
    \underset{h_{-1}}{\overset{h_1}{\rightleftarrows}}
    {S_2F}
    \overset{h_2}{\rightarrow}
    {S_1F}^*
    \underset{h_{-3}}{\overset{h_3}{\rightleftarrows}}
    {S_1+F}
    \underset{h_{-4}}{\overset{h_4}{\rightleftarrows}}
    {S_1F}
    \overset{h_5}{\rightarrow}
    {S_0F}
    \underset{h_{-6}}{\overset{h_6}{\rightleftarrows}}
    {S_0} + {F}.\\
\end{split}
\end{equation}
The first  component features a strongly irreversible mechanism, since no reaction is assumed from $S_2+E$ to $S_1E$ (nor from $S_1+E$ to $S_0E$).
On the other side, the enzyme $F$ is assumed to rebind to $S_1$ and $S_0$ and is thus a weakly irreversible mechanism. 
\end{remark}

 We will not recall the general notion of {\em complete binomial network} introduced in~\cite{SFeliu} because it cannot be checked in general. Instead, we introduce a particular class of complete binomial networks which can be algorithmically checked in an easy way and wec will show many interesting examples. We will also replace their definition of {\em surjectivity} by the basic notion of consistency in Definition~\ref{def:consistent}, as it is already remarked in~\cite{PMDSC,SFeliu}.

\begin{definition} \label{def:lebn}
Let $G$ be a network with $r$ reactions and $s$ species and let $\dot x = f(x)$ 
denote the resulting mass-action system. Assume without loss of generality that $f_{1}, \dots, f_{s-d}$ is a maximal set of linearly independent polynomials among $f_1, \dots, f_s$.
We say $G$ is \emph{linearly equivalent to a binomial network} (denoted by \emph{lebn}) if there exist binomials $h_{1},  \dots, h_{s-d}$ in $\mathbb{Q}(\kappa)[x]$ with two nonzero terms with coefficients of opposite sign, and an $(s-d)\times (s-d)$ matrix $M(\kappa)$ with entries in $\mathbb{Q}(\kappa):=\mathbb{Q}(\kappa_1,\kappa_2,\dots, \kappa_r)$ defined and invertible for all $\kappa \in \mathbb{R}^r_{>0}$, such that the binomials $(h_{j })$ can be obtained  as follows:
\begin{equation}\label{eq:linearchange}
(h_{1}, \dots, h_{s-d})^{\top}
\quad := \quad M(\kappa) ~  (f_{1},\dots, f_{s-d})^{\top}~.
\end{equation}
 \end{definition}

If a binomial $h_j$ has just one nonzero term  or two terms with the same sign, then there are no positive steady states.

\begin{example}[Two-layer cascade]\label{ex:cascade2}
Consider the cascade motif from~\cite[Figure 1(k)]{FW2012} with two layers. Each layer is a one-site modification cycle, and the same phosphatase ($F$) acts in each layer, usually depicted as in Figure~\ref{fig:SP}. 
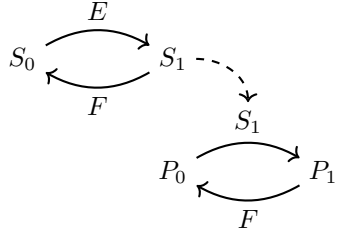
\begin{figure}[h!]
\begin{tikzpicture}[xscale=1,yscale=1]
 \node (S_0) at (0, 0) {$S_0$};
 \node (S_1) at (2, 0) {$S_1$};
 \node (P_0) at (2, -1.5) {$P_0$};
 \node (P_1) at (4, -1.5) {$P_1$};
 \node (P_01) at (3,-0.7) {};
  \draw
    (S_0) edge[->,thick,bend left] node [above] {$E$} (S_1)
       (S_1) edge[->,thick,bend left] node [below] {$F$} (S_0)
       (P_0) edge[->,thick,bend left] node [above] {$S_1$} (P_1)
    (P_1) edge[->,thick,bend left] node [below] {$F$} (P_0);
    \draw[dashed, ->,thick,looseness=1] (S_1) to [out=0,in=90] (P_01) ;   
\end{tikzpicture}
\caption{Two-layer cascade}
\label{fig:SP}
\end{figure}
Making explicit the usual intermediate species and naming the rate constants we get the following reaction network:
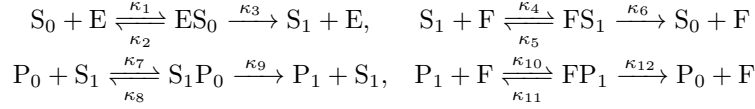
\begin{center}
\begin{tabular}{cc}
\ce{S_{0} + E
<=>[\kappa_1][\kappa_2]
ES_{0}
->[\kappa_3]
S_{1} + E},  &
\ce{S_{1} + F
<=>[\kappa_4][\kappa_5]
FS_{1}
->[\kappa_6]
S_{0} + F}\\
\ce{P_{0} + S_{1}
<=>[\kappa_7][\kappa_8]
S_{1}P_{0}
->[\kappa_9]
P_{1} + S_{1} }, 
&
\ce{P_{1} + F
<=>[\kappa_{10}][\kappa_{11}]
FP_{1}
->[\kappa_{12}]
P_{0} + F} \\
\end{tabular}
\end{center}
We order the species as $X_1$=\ce{S_0}, $X_2$=\ce{ES_0}, $X_3$=\ce{S_1}, $X_4$=\ce{FS_1}, $X_5$=\ce{P_0}, $X_6$=\ce{S_1P_0}, $X_7$=\ce{P_1}, $X_8$=\ce{FP_1}, $X_9$=\ce{E}, $X_{10}$=\ce{F}. There are $d=4$ conservation laws, which arise from the total amounts of substrates and enzymes:
\begin{eqnarray}\label{eq:ex45con}
x_1 + x_2 + x_3 + x_4 + x_6 &= c_1, \quad \quad \quad \quad\quad \ 
x_2 + x_9 &= c_2,\\
x_4 + x_8 + x_{10} &= c_3, \quad  
x_5 + x_6 + x_7 + x_8 &= c_4 \notag.
\end{eqnarray}
From these conservation laws, we see that this network is conservative. 
We consider the following maximally linearly independent steady state polynomials:
\begin{small}
\begin{align*}
f_{3} ~&=~  \kappa_3x_2 - \kappa_4x_3x_{10} + \kappa_5 x_4 - \kappa_7x_3x_5 + \kappa_8x_6 + \kappa_9x_6, \\
f_{6} ~&=~ \kappa_7x_3x_5 - \kappa_8x_6 - \kappa_9x_6, \\
f_{7} ~&=~ -\kappa_{10}x_7x_{10} + \kappa_{11}x_8 + \kappa_{9}x_6,  \\
f_{8}~&=~ \kappa_{10}x_7x_{10}  - \kappa_{11} x_8 - \kappa_{12} x_8, \\
f_{9}~&=~  -\kappa_1x_1x_9  +\kappa_2x_2  + \kappa_3x_2 ,  \\
f_{10} ~&=~ -\kappa_4x_3x_{10} +\kappa_5x_4 +\kappa_6x_4 -\kappa_{10}x_7x_{10}  + \kappa_{11} x_8 + \kappa_{12} x_8 .  
\end{align*}
\end{small}
Multiplying this system by the following matrix with positive determinant:
{\footnotesize
\[M(\kappa)=
\left(\begin{array}{cccccccccc}
 1 &   1 & 0 & -1 & 0 & -1\\
 0 &   1 & 0 & 0  & 0 & 0\\
 0 & 0 &  1 &   1 & 0 & 0\\
 0    &  0 & 0     &  1 & 0 & 0\\
 0    &  0 & 0   & 0   &  1 & 0\\
 0    & 0 & 0   & 1         & 0 &  1
\end{array}
\right),
\]
} 
we obtain the binomial system:
\begin{small}
\begin{equation}\label{eq:cascade_binom}
\begin{array}{ll}
h_{3}  =  \kappa_3x_2 - \kappa_6x_4, &
h_{6}= \kappa_7x_3x_5 - (\kappa_8 + \kappa_9)x_6, \\
h_{7} =\kappa_9x_6-\kappa_{12} x_8,  &
h_{8} = \kappa_{10}x_7x_{10}  -(\kappa_{11} + \kappa_{12}) x_8,\\ 
h_{9}  =  -\kappa_1x_1x_9  + (\kappa_{2} + \kappa_{3})x_2, &
h_{10}  = -\kappa_4x_3x_{10} +(\kappa_{5} + \kappa_{6})x_4,
\end{array}
\end{equation}
\end{small}
and we see that the network is lebn.  This  will also follow from Theorem~\ref{th:ylebn}.
\end{example}

An important practical consequence of the notion of a lebn network is that this is an easily checkable condition. Given a network $G$ and the associated polynomials $f_1, \dots, f_n$, we order the set of exponents occurring in them (that is, the source complexes) and consider the associated matrix $C_f$ of corresponding coefficients. Let $R_f$ be the unique reduced row echelon form  of the matrix $C_f$ and $M(\kappa)$ in $\Q(\kappa)$ be the invertible matrix yielding the equality $R_f =  M(\kappa) \, C_f$.

\begin{proposition} \label{prop:CK}
 If a network $G$ is lebn, then the reduced row echelon form $R_f$ of the matrix $C_f$ has two non-zero entries with different signs on each row. The matrix $M(\kappa)$ has rational form entries whose denominators are non-vanishing on the positive orthant. 
 
 \par
 
 Conversely, if the reduced row echelon form $R_f$ of $C_f$ has two non-zero entries  with different signs on each row, it is enough to check if the invertible matrix $M(\kappa)$  has rational form entries whose denominators are non-vanishing on the positive orthant to ensure that  $G$ is lebn. 
\end{proposition}

The first assertion in Proposition~\ref{prop:CK} follows from \citetext{Prop.~2.2; \citealp{CK}} and \cite[Thm. 3.3]{PMDSC} and the second assertion is straightforward.

\smallskip

The property of being lebn can be lifted to non-confluent extensions. Recall that $T$ is the map defined in~\eqref{eq:T}. 

\begin{proposition}\label{prop:lebn}
Given a lebn network $G'$ with rate constants $\tau$, any non-confluent extension $G$ of $G'$ with rate constants $\kappa$ such that $T(\kappa)=\tau$  is also lebn. 
In particular, all canonical extensions of a lebn network are lebn.
\end{proposition}

\begin{proof}
Let $G$ be a non-confluent extension of $G'$ with reactions constants $\tau=T(\kappa)$. We order the species of $G$ as $X_1,\dots, X_n,U_1,\dots,U_p$. If $g_{j_1},\dots,g_{j_{n-d}}$ is a basis of the $\Q(\tau)$-linear subspace generated by $g_1,\dots,g_n$, it follows from Lemma~\ref{lem:linear} that $f_{j_1},\dots,f_{j_{n-d}},f_{n+1},\dots,f_{n+p}$ generate the $\Q(\kappa)$-linear subspace spanned by $f_1,\dots,f_s$. As already observed in Remark~\ref{rem:uk}, the equations $v_1,\dots, v_p$ in Lemma~\ref{lem:linear} are binomials of the form $u_i - \mu_i \, x^{y(i)}$. We moreover know by Lemma~\ref{lem:linear} that after linearly eliminating all intermediate species,  we get the  equations of 
$G'$ with reactions constants $\tau=T(\kappa)$.  Then there is an $(s-d)\times (s-d)$ invertible matrix $M_1$ with entries in $\mathbb{Q}(\kappa)$ defined for all $\kappa \in \mathbb{R}^r_{>0}$ such that 
$$M_1 \, (f_{j_1},\dots,f_{j_{n-d}},f_{n+1},\dots,f_{n+p})^{\top}=(g_{j_1},\dots,g_{j_{n-d}},v_1,\dots,v_p)^{\top}.$$ As $G'$ is lebn and $\tau=T(\kappa)$,  there exist a set of $n-d$ binomials $\{h_{j_1},\dots,h_{j_{n-d}}\}$  with two nonzero terms with coefficients of opposite sign and an invertible matrix $M_2\in \Q(\kappa)^{(n-d)\times(n-d)}$, which is well defined for all $\kappa \in \mathbb{R}^r_{>0}$, such that $M_2 \,(g_{j_1},\dots,g_{j_{n-d}})^{\top}=(h_{j_1},\dots,h_{j_{n-d}})^{\top}$. Call $M \in \Q(\kappa)^{(s-d)\times(s-d)}$ the matrix obtained as the product of $M_1$ with following block matrix:
\[
M \, = \, \begin{pmatrix}M_2 &0\\0& Id_p\end{pmatrix} \, M_1,
\]
where $Id_p$ is the $p\times p$ identity matrix, Then, $M$ is invertible and well defined for all $\kappa \in \mathbb{R}^r_{>0}$ and, as we wanted to prove, it holds that 
\[M \, (f_{j_1},\dots,f_{j_{n-d}},f_{n+1},\dots,f_{n+p})^{\top}=(h_{j_1},\dots,h_{j_{n-d}},v_1,\dots,v_p)^{\top}.
\]
 \end{proof}

Although we are not going to introduce the definition of a complete binomial network from~\cite{SFeliu}, the purpose of introducing the algorithmically checkable notions of lebn and consistent networks is explained in the following remark.

\begin{remark}\label{rmk:l+c=>cb}
 All the references in this remark are from~\cite{SFeliu}. A {\em complete binomial network} is characterized in Definition~2.7. It should admit an {\em admissible binomial basis} (introduced in their Definition~2.1), which satisfies the conditions (rank) and (surj) given before Definition~2.7. It is straightforward to verify that a lebn network has an admissible binomial basis that satisfies (rank). The equivalence between (surj) and consistency is proven  in Lemma~2.6 in that article. Therefore, any consistent lebn network is a complete binomial network.
\end{remark}

Given a lebn (or more generally a complete binomial) network  $G$  with associated dynamical system $\dot{x}=f(x)$, the positive steady states (that is, the positive solutions of $f_1(x) = \dots = f_s(x)=0$) coincide with the positive solutions of a set of binomials with coefficients which are rational functions on the rate constants (well defined over the positive orthant). As in Proposition~\ref{prop:CK}, if  the system does have positive steady states then the coefficients of these binomials must have different signs. Moreover, the positive zeros of such a binomial coincide with the positive solutions of an equation of the form $x^m = \gamma(\kappa)$ with $\gamma(\kappa) > 0$ and $x^m$ a {\em Laurent} monomial (that is, a monomial with {\em integer} exponents $m \in \Z^s$). We now define a polynomial associated to a set of monomials and linear forms (that will be chosen to be a basis of the linear conservation relations).

\begin{definition}\label{def:Bpoly}
Given Laurent monomials $x^{m_i}$ in $s$ variables , $i = 1, \dots, s-d$, and homogeneous linear forms $\ell_1 = \sum_{j=1}^s w_{1j} x_j, \dots, \ell_d = \sum_{j=1}^s w_{dj} x_j$, we consider the following matrices. We denote by $\mathrm{Exp} \in \Z^{s \times (s-d)}$  the matrix with columns $m_1, \dots, m_{s-d}$ and  by $Wx$ the matrix of linear forms:
\begin{equation}\label{eq:Exp}
 Wx=\begin{pmatrix}w_{11} x_1 & \dots & w_{1s} x_s \\ \vdots & \dots & \vdots\\ w_{d1} x_1 & \dots & w_{ds} x_s\end{pmatrix}.
\end{equation}
We define the matrix $J_{\mathrm{Exp}^t,Wx}$ with first $s-d$ rows equal to the transpose matrix  $\mathrm{Exp}^t$ and last rows equal to $Wx$.
We denote by $B$ the homogeneous polynomial
\begin{equation}\label{eq:Bpoly}
B(x) \, = \det J_{\mathrm{Exp}^t,W\!x}.
\end{equation}
\end{definition}

Note that if $B$ is not the zero polynomial then it has degree $d$  in $x=(x_1, \dots, x_s)$ and all its exponents are either $0$ or $1$.

\begin{example}[Example~\ref{ex:cascade2} continued]\label{ex:cascadeB}
Recall the two-layer cascade from Example~\ref{ex:cascade2}. We can read from~\eqref{eq:ex45con} the matrix $W\!x$ and from the binomials in~\eqref{eq:cascade_binom} we obtain the Laurent monomials $x_2x_4^{-1}$, $x_3x_5x_6^{-1}$, $x_6x_8^{-1}$, $x_7x_{10}x_8^{-1}$, $x_1x_9x_2^{-1}$ and $x_3x_{10}x_4^{-1}$ to build the matrix $\mathrm{Exp}^t$. We then have
$$
  B(x) = \det \begin{pmatrix} 
 0 &  1 & 0 & -1 & 0 &  0 & 0 &  0 & 0 & 0\\ 
 0 &  0 & 1 &  0 & 1 & -1 & 0 &  0 & 0 & 0\\ 
 0 &  0 & 0 &  0 & 0 &  1 & 0 & -1 & 0 & 0\\ 
 0 &  0 & 0 &  0 & 0 &  0 & 1 & -1 & 0 & 1\\ 
 1 & -1 & 0 &  0 & 0 &  0 & 0 &  0 & 1 & 0\\ 
 0 &  0 & 1 & -1 & 0 &  0 & 0 &  0 & 0 & 1 \\ \hline
 x_1 & x_2 & x_3 & x_4 & 0 & x_6 & 0 & 0 & 0 & 0 \\ 
 0 & x_2 & 0 & 0 & 0 & 0 & 0 & 0 & x_9 & 0 \\
 0 & 0 & 0 & x_4 & 0 & 0 & 0 &  x_8 & 0 & x_{10} \\
 0 & 0 & 0 & 0 & x_5 & x_6 & x_7 & x_8 & 0 & 0 
 \end{pmatrix}
$$

\begin{align*}
& = -x_1x_2x_5x_8 +x_1x_2x_5x_{10} +x_1x_2x_6x_{10} +x_1x_2x_7x_8 +x_1x_2x_7x_{10} +x_1x_2x_8x_{10} \\
& -x_1x_5x_8x_9 +x_1x_5x_9x_{10} +x_1x_6x_9x_{10} +x_1x_7x_8x_9 +x_1x_7x_9x_{10} +x_1x_8x_9x_{10} \\
& -x_2x_5x_8x_9 +x_2x_5x_9x_{10} +x_2x_6x_9x_{10} +x_2x_7x_8x_9 +x_2x_7x_9x_{10} +x_2x_8x_9x_{10} \\
& +x_3x_4x_5x_9 +x_3x_4x_6x_9 +x_3x_4x_7x_9 +x_3x_4x_8x_9 +x_3x_5x_9x_{10} +x_3x_6x_9x_{10} \\
& +x_3x_7x_8x_9 +x_3x_7x_9x_{10} +x_3x_8x_9x_{10} +x_4x_5x_6x_9 -x_4x_5x_8x_9 
+x_4x_5x_9x_{10} \\
& -x_4x_6x_7x_9 +x_4x_6x_9x_{10} +x_4x_7x_8x_9 +x_4x_7x_9x_{10} +x_4x_8x_9x_{10} +x_5x_6x_9x_{10}.
\end{align*}

\end{example}

The following lemma is a restatement of~\cite[Lemma~2.10]{Focm}. Given a matrix $K$ with $s$ columns and a subset $I \subset \{1, \dots,s\}$, we denote by $K_I$ the submatrix of $K$ consisting of the columns of $K$ with indices in $I$ (in the same order).

\begin{lemma} \label{lem:Laplace}
 Let $m_1, \dots, m_{s-d} \in \Z^s$ and linear forms $\ell_1, \dots, \ell_d$ as in Definition~\ref{def:Bpoly}. Then for any subset $I$ of $\{1, \dots, s\}$ with cardinality $d$, the coefficient of the monomial $\prod_{ i \in I} x_i$ in the polynomial $B$ in \eqref{eq:Bpoly} equals --up to sign-- the product of the determinant of the submatrix $W\!x_I$  times the determinant of the submatrix ${\rm Exp}^t_{I^c}$ indicated by the columns not in $I$.
\end{lemma}

The proof of Lemma~\ref{lem:Laplace} is a direct consequence of the Laplace expansion of the determinant by complementary minors.

\begin{corollary} \label{cor:Bnon0} 
In the hypotheses of Lemma~\ref{lem:Laplace}, $B$ is not the zero polynomial if and only if there exists $I \subset \{1, \dots, s\}$ with cardinality $d$ such that
\begin{equation}\label{eq:dets}
 \det(W_I) \det({\rm Exp}^t_{I^c}) \not= 0.
\end{equation}
 
 \end{corollary}

 The proof of Corollary~\ref{cor:Bnon0} is a straightforward consequence of Lemma~\ref{lem:Laplace} because the matrix $W\!x$ equals the product of the conservation-law matrix $W$ and the invertible diagonal matrix with entries $(x_1, \dots, x_s)$.

\medskip
The definition of the polynomial $B$ in~\eqref{eq:Bpoly} is motivated by the following proposition that permeates many different results in the literature; we refer in particular to Section~2 in~\cite{Focm} where several of the previous results were unified.

\begin{proposition} \label{prop:Bpoly}
 Let $m_1, \dots, m_{s-d} \in \Z^s$ and linear forms $\ell_1, \dots, \ell_d$ as in Definition~\ref{def:Bpoly}. Then, there exist two different points $x, y \in \R^s_{>0}$ satisfying
 \[ x^{m_i} = y^{m_i}, \, \text{ for } i=1, \dots, s-d, \quad \ell_j(x) = \ell_j(y), \, \text{ for } j =1, \dots,d,
       \]
if and only if either $B$ is the zero polynomial or it has two coefficients with different signs.
 \end{proposition}

\medskip

Assume we have binomials $h_1,\dots, h_{s-d}$ with coefficients of different signs defining the positive steady states of the dynamical system associated to a reaction network. The positive zeros of each $h_i$ coincide with the positive solutions of an equation $x^{m_i} - \gamma_i(\kappa) =0$ with $\gamma_i(\kappa) > 0$, but one can also choose the monomial $x^{-m_i}$ (and right hand side $\gamma_i(\kappa)^{-1}$). Also, let $\ell_1, \dots, \ell_d$ be a choice of a basis of the linear conservation relations of the network. For any lebn network, both the binomials and the linear forms are defined up to multiplication by an invertible matrix not depending on the $x$ variables, so we can associate to the network a polynomial $B$ as in \eqref{eq:Bpoly}, defined up to non-zero constant.

Theorem~\ref{th:Bpoly} below is essentially a restatement of Lemma~4.7, Lemma~4.8 and Algorithm~4.9 in~\cite{SFeliu}. We only write it for the case of lebn networks, but it could be stated for complete binomial networks. Our variables $x_1, \dots, x_s$ are related to the variables $\lambda_1, \dots, \lambda_s$ in~\cite{SFeliu} via $x_i = (\lambda_i)^{-1}$.  With our definitions, a basic step in the proof is the following result. Given a complete binomial network $G'$ with concentrations $x=(x_1, \dots, x_{s'})$ and a non-confluent extension $G$ of $G'$ given by the addition of a single intermediate species (whose concentration
we denote by $u$), $G$ is also complete binomial. Moreover, if $G'$ is lebn we have by Proposition~\ref{prop:lebn} that $G$ is also lebn. 

Consider a basis of linear conservation relations for $G'$ and their extension to a basis of linear conservation relations as in~\eqref{eq:cons_laws_G}. We can {\it linearly} add to the binomials describing the positive steady states of $G'$ a binomial of the form $u- \mu x^y$ as in~\eqref{eq:ukbinom}. Then, we get the following relation between the corresponding polynomials $B_{G'}$ and $B_G$:
 \begin{equation}\label{eq:Bxu}
  B_G(x,u) = \pm \, B_{G'}(x)+u \, P(x),
 \end{equation}
where $P$ is a polynomial that only depends on $x$. In particular, setting $u=0$ in $B_G$ gives $B_{G'}$.

\begin{theorem}\label{th:Bpoly}
 Let $G'$ be a lebn network that is not multistationary. Let $C\subseteq \mathcal C_{G'}$ be a subset of complexes such that the canonical extension $G=G'_{ext,C}$ associated to $C$ is multistationary. Let $h_1, \dots, h_{s'-d}$ be binomials as in Definition~\ref{def:lebn} for $G'$.  Add binomials $h_{s'-d+1}, \dots, h_s$ of the form ~\eqref{eq:ukbinom} as in Proposition~\ref{prop:lebn}. Take a basis of linear conservation relations for $G'$ and consider their extension to a basis of linear conservation relations for $G$ as in~\eqref{eq:cons_laws_G}. We call $x$ the vector of concentrations of the species in $G'$ and $u$ the vector of concentrations of the intermediate species added in $G$. 
 \par
 Let $B(x,u)$ denote the associated polynomial of $G$ in Definition~\ref{def:Bpoly},  and write 
 \begin{equation}\label{B_gamma}
  B(x,u)=\underset{\upsilon=(\upsilon',\upsilon'')}{\sum} b_\upsilon~x^{\upsilon'}u^{\upsilon''}.
 \end{equation}
 
 Then:
 \begin{itemize}
  \item  The sign $\sigma$ of the all the terms in $B(x,0)$ is non-zero and for any fixed index $\upsilon=(\upsilon',\upsilon'')$ such that $\mathrm{sign}(b_\upsilon)=-\sigma$, 
 the subnetwork of $G$ obtained by elimination of the subset of intermediates $\{U_i~:~ \upsilon''_i= 0\}$ exhibits multistationarity.
 \item Let
 $\Upsilon =\{\upsilon''\in\{0,1\}^p :\text{ there exists }\upsilon=(\upsilon',\upsilon'')
 \text{ with } \mathrm{sign}(b_\upsilon)= -\sigma\}$.
 There is a bijection between the elements of this set with minimal support and the circuits of multistationarity associated to $C$. 
 \item
  Call $y_i \in C$ the input complex of the intermediate $U_i$. For each $\upsilon'' \in \Upsilon$ with minimal support, the corresponding circuit of multistationarity is $\{y_i:~\upsilon_i''\neq 0\}$. 
  \end{itemize}
\end{theorem}

 Note that $B(x,0)$ is the associated polynomial of $G'$ by equality~\ref{eq:Bxu}. Then as $G'$ is not multistationary, it follows that $B(x,0)$ cannot be identically zero by Proposition~\ref{prop:Bpoly} and then $B(x,u) \not\equiv 0$.

\begin{remark}
  We show in Section~\ref{sec:Messi} how to check for monostationarity in the reduced network by visual inspection of associated graphs in the case $G$ is a MESSI network \cite{messi}. The two-layer cascade (see Example~\ref{ex:cascade2}), the ERK signaling pathway (see Section~\ref{sec:ERK}) and the sequential phospho/dephosphorylation networks (see Section~\ref{sec:sequential}) are examples of MESSI networks that verify the conditions to have a monostationary core network.
\end{remark}

\begin{remark}
 In the case where $G$ has no relevant boundary steady states (and consequently $G_{red}$ has no relevant boundary steady states if $G$ is non-confluent, see Proposition~\ref{prop:bss}), the polynomial $B$ can be obtained from a {\it critical function} \cite[Lemma~4.4]{DPMST} and the multistationarity results can be recovered using degree theory (see also~\cite{CFMW}). One advantage of this approach is that the system is multistationary if and only if the polynomial $B$ has a coefficient with sign $(-1)^{\dim(\mathcal{S})+1}$ \cite[Thm.~4.6]{DPMST}. If moreover $G$ is consistent, every $x^*\in\R^s_{>0}$ with $\mathrm{sign}(B(x^*))= (-1)^{\dim(\mathcal{S})+1}$ yields a witness to multistationarity $(\kappa^*,c^*)$ by defining $c^*=Wx^*$ (with $W$ the conservation-law matrix and $\kappa^*$ such that $f_{\kappa^*}(x^*)=0$). We expand on this approach in Example~\ref{ex:cascade_cont} below.
\end{remark}

\begin{example}[Example~\ref{ex:cascadeB}, continued]\label{ex:cascade_cont}

Recall the polynomial $B$ from Example~\ref{ex:cascadeB} for the two-layer cascade. By setting $x_2=x_4=x_6=x_8=0$ (i.e. by deleting all the terms where an intermediate species is involved) it is easy to check that we obtain a polynomial with positive coefficients only. If we inspect the negative terms in $B$, we obtain $\Upsilon =\{(1,0,0,1),(0,0,0,1),(0,1,0,1),(0,1,1,0)\}$ (the coordinates correspond to the intermediate species $X_2$, $X_4$, $X_6$ and $X_8$, respectively). The vectors in $\Upsilon$ with minimal support are $(0,0,0,1)$ and $(0,1,1,0)$. The set of inputs of the network is $\{X_1+X_9,X_3+X_{10},X_5+X_3,X_7+X_{10}\}$ and then the circuits of multistationarity associated to these inputs are the sets $\{X_7+X_{10}\}$ and $\{X_3+X_{10},X_5+X_3\}$. This is, $\{P_1+F\}$ and $\{S_1+F, P_0+S_1\}$.

We can go further and find a witness to multistationarity as it can be shown that this system is consistent and does not have relevant boundary steady states \cite{messi}. Then, by  \cite[Thm.~4.6]{DPMST} we know that any $x^*\in\R^s_{>0}$ with $\mathrm{sign}(B(x^*))= (-1)^{6+1}$ yields a witness to multistationarity. In order to find it, define for instance $x^*\in {\mathbb R}^{10}$ with coordinates $\lambda$ (in indices $1,2,5,$ and $8$) and 1 (all others):
 \begin{align} \label{eq:x-lambda}
 x^*~=~
  \left(
   \lambda,~  \lambda, ~ 1, ~ 1, ~ \lambda, ~ 1, ~ 1, ~ \lambda, ~ 1,~  1
  \right).\end{align}
 After replacing we obtain
 \begin{align} \label{eq:B-phos}
 B(x^*)~=~
 -\lambda^4+\lambda^3 + 7\lambda^2 + 14\lambda + 5.\end{align}
 It follows that $B(x^*)<0$ if $\lambda$ is larger than the largest positive root of the polynomial~\eqref{eq:B-phos}.  We can use an elementary bound for this root, for instance  the sum of the absolute values of all coefficients: $1+1+7+14+5 = 28$. Let $\lambda = 29$; then, $B(x^*)|_{\lambda=29}=-676594<0$. To solve for $c^*$, we substitute  $x^*|_{\lambda=29}$, as in~\eqref{eq:x-lambda}, into equation~\eqref{eq:ex45con}, which yields:
 \begin{align*}
 c^* ~=~ (61, 30, 31, 60).
 \end{align*}
 Finally, we choose $\kappa^*$ for which $ f_{\kappa^*} (x^*)=0$:
  \begin{align*}
 \kappa^* ~=~ (2,~1,~1,~30,~1,~29,~1,~28,~1,~2,~ 1/29,~ 1/29)~.
 \end{align*}
 So, $(\kappa^*, c^*)$ is a witness to multistationarity.   We can numerically approximate other positive steady states in the same
 stoichiometric compatibility class. Using the command {\rm RootFinding[Isolate]} in Maple we found the values
 {\small\[x^{**} = (30.833,29.058,0.036,1.003,57.85, 0.072,0.003,2.076,0.942,27.922),\]
 \[x^{***} = (7.883, 26.623, 24.541, 0.918, 1.224, 1.036, 27.695,  30.045, 3.377,  0.037).
 \]

Note that the seventh coordinate on each steady state, which corresponds to the concentration of $P_1$, is very small in $x^{**}$ and big in $x^{***}$ and the opposite happens with the values of $S_1$.
}

\end{example}

\section{ERK pathway}\label{sec:ERK}

The ERK cascade is primarily segmented into three stages: RAS activation, succeeded by the MAPK kinase kinase (MAPKKK) or RAF, then by the MAPK kinase (MAPKK) or MEK, and lastly the MAPK or ERK. Each enzyme in this sequence activates the subsequent one through a series of phosphorylation events. MAPKKKs phosphorylate MAPKKs at two conserved serine residues, and MAPKKs subsequently phosphorylate MAPKs at conserved threonine and tyrosine residues \cite{sig-016}. The culmination of this cascade is marked by the activation of ERK. Once activated, ERK possesses the ability to move into the cell's nucleus. Within the nucleus, ERK assumes a critical role, influencing the function of numerous transcription factors. This regulation eventually affects particular gene expressions, inciting changes in the cell's behavior in response to the original external prompt.

In Examples~\ref{ex:3layercascade} and~\ref{ex:3layercascadeloop} that follow, we compute the circuits of multistationarity for both the Three-layer cascade and the Three-layer cascade with a feedback loop. One can easily verify that these networks are conservative and have no boundary steady states using the criteria in \cite{messi}. We adopt a simplified notation that deviates from conventional biochemistry nomenclature. Here, $E$ represents the RAS activation, $R_2$ stands for the final output ERK of the cascade, and $F_1,F_2$, and $F_3$ denote the phosphatases in each respective layer (which might be identical or distinct), among other notations.

Moreover, since the computations become much longer with the critical polynomial $B(x,u)$ in Theorem~\ref{th:Bpoly} having hundreds or even thousands of terms we do not explicit all computations anymore. Instead, we use a \cite{maple} script that can be downloaded from~\url{http://mate.dm.uba.ar/~alidick/ERKFiles}. The script does what we have done in Example~\ref{ex:cascade_cont} in a automated fashion, one just needs to insert the core reactions. It can be easily changed to work for other networks from small to medium size. We used this approach for networks of about 30 parameters/variables. The computations for all examples in the whole manuscript together took less than one minute to run in a notebook with 16MB of RAM and with 8 cores of 2.6GHz of processing.

\begin{example}[Three layer ERK cascade]\label{ex:3layercascade}
Consider the network in Figure~\ref{fig:SPR} from \cite[Figure 1]{sig-016}.
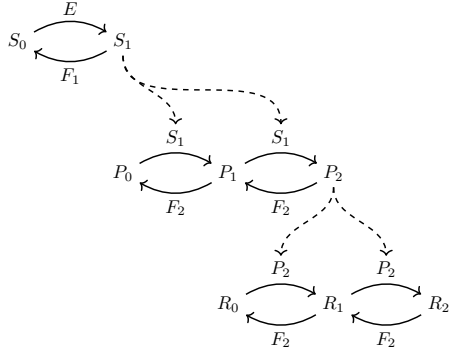
\begin{figure}[h!] 
\scalebox{0.7}{
\begin{tikzpicture}[xscale=1,yscale=1]
 \node (S_0) at (0, 0) {$S_0$};
 \node (S_1) at (2, 0) {$S_1$};
 \node (P_0) at (2, -2.5) {$P_0$};
 \node (P_1) at (4, -2.5) {$P_1$};
 \node (P_2) at (6, -2.5) {$P_2$};
 \node (R_0) at (4, -5) {$R_0$};
 \node (R_1) at (6, -5) {$R_1$};
 \node (R_2) at (8, -5) {$R_2$};
 \node (P_01) at (3,-1.7) {};
 \node (P_12) at (5,-1.7) {};
 \node (R_01) at (5,-4.2) {};
 \node (R_12) at (7,-4.2) {};
  \draw
    (S_0) edge[->,thick,bend left] node [above] {$E$} (S_1)
       (S_1) edge[->,thick,bend left] node [below] {$F_1$} (S_0)
    (P_0) edge[->,thick,bend left] node [above] {$S_1$} (P_1)
    (P_1) edge[->,thick,bend left] node [above] {$S_1$} (P_2)
    (P_1) edge[->,thick,bend left] node [below] {$F_2$} (P_0)
    (P_2) edge[->,thick,bend left] node [below] {$F_2$} (P_1)
    (R_0) edge[->,thick,bend left] node [above] {$P_2$} (R_1)
    (R_1) edge[->,thick,bend left] node [above] {$P_2$} (R_2)
    (R_1) edge[->,thick,bend left] node [below] {$F_2$} (R_0)
    (R_2) edge[->,thick,bend left] node [below] {$F_2$} (R_1);
\draw[dashed, ->,thick,looseness=1] (S_1) to [out=-90,in=90] (P_01) ;    
\draw[dashed, ->,thick,looseness=1] (S_1) to [out=-90,in=90] (P_12) ; 
\draw[dashed, ->,thick,looseness=1] (P_2) to [out=-90,in=90] (R_01) ;    
\draw[dashed, ->,thick,looseness=1] (P_2) to [out=-90,in=90] (R_12) ;  
\end{tikzpicture}
}
\caption{Three-layer ERK cascade}
\label{fig:SPR}
\end{figure}

\noindent
This network is a cascade motif with three layers;
the first layer is a one-site modification cycle and the second and third layers are two-site modification cycles. The same phosphatase ($F_2$) acts in the last two layers but is different from the phosphatase acting on the first layer ($F_1$).  

Following the same procedure as in 
Example~\ref{ex:cascade_cont} we obtain a critical function 
$B(x,u)$ which is a polynomial with $1334$ monomials. Setting all concentrations of 
intermediates equal to zero we get the critical function $B(x,0)$ corresponding to the core system 
whose  monomials have a fixed sign $\sigma \in \{+,-\}$.

In order to apply Theorem~\ref{th:Bpoly}, we look at monomials in $B(x,u)-B(x,0)$  with sign different from $\sigma$. There are $458$ such monomials with the {\em wrong} sign, among which we consider only those monomials which depend on the concentrations of intermediates  species and we look for the monomials with minimal support. The set of circuits of multistationarity is $$ 
\{ \{P_0+S_1\}, \{P_2+F_2\}, \{R_0+P_2\}
, \{R_1+F_2\}, \{R_2+F_2\}, \{P_1+F_2, R_1+P_2\} \}.$$

We can similarly
compute the circuits of multistationarity of variations of this network  considering equal phosphatases to act 
in each layer, or a different one in each layer, etc. We summarize the results obtained in Table~\ref{table:3layernoloop}.

\begin{table}[!htbp]
\centering
\begin{small}
\begin{tabular}{c|c}
Enzymes & Canonical Circuits of Multistationarity\\
\hline
\multirow{2}{*}{$F_1,F_1,F_1$}
& $ \{P_0+S_1\}, \{P_1+F_1\},
  \{P_2+F_1\}, \{R_0+P_2\},
  \{R_1+F_1\},\{R_2+F_1\},  $\\ 
& $ 
\{P_1+S_1, S_1+F_1\}$\\
\hline
\multirow{2}{*}{$F_1,F_1,F_3$}
& 
$ \{P_0+S_1\}, \{P_1+F_1\},
  \{P_2+F_1\}, \{R_0+P_2\},
  \{R_2+F_3\}, $\\ 
& $  \{P_1+S_1, S_1+F_1\} $ \\
\hline
\multirow{2}{*}{$F_1,F_2,F_1$}
& $ \{P_0+S_1\}, \{P_2+F_2\}, \{R_0+P_2\},
  \{R_1+F_1\}, \{R_2+F_1\},  $\\ 
& $    
 \{P_1+S_1, R_1+P_2, S_1+F_1\}$ \\
\hline
\multirow{2}{*}{$F_1,F_2,F_2$}
& 
$ \{P_0+S_1\}, \{P_2+F_2\}, \{R_0+P_2\}
, \{R_1+F_2\}, \{R_2+F_2\}, $\\ 
& $ \{P_1+F_2, R_1+P_2\} $ \\
\hline
\multirow{1}{*}{$F_1,F_2,F_3$}
& 
$ \{P_0+S_1\}, \{P_2+F_2\}, \{R_0+P_2\}
, \{R_2+F_3\} $\vspace{0.3cm} 
\end{tabular}
\end{small}
\caption{Canonical circuits of multistationarity for the three-layer cascade. The first column shows which phosphatase acts on each layer.
}\label{table:3layernoloop}
\end{table}

\end{example}

\begin{example}[Three-layer  ERK cascade with feedback loop]\label{ex:3layercascadeloop}
Consider now the network in Figure~\ref{fig:3layerloop} which incorporates a negative feedback loop.
\begin{figure}
\scalebox{0.7}{
\begin{tikzpicture}[xscale=1,yscale=1]
 \node (S_0) at (0, 0) {$S_0$};
 \node (S_1) at (2, 0) {$S_1$};
 \node (S_2) at (4, 0) {$S_2$};
 \node (P_0) at (2, -2.5) {$P_0$};
 \node (P_1) at (4, -2.5) {$P_1$};
 \node (P_2) at (6, -2.5) {$P_2$};
 \node (R_0) at (4, -5) {$R_0$};
 \node (R_1) at (6, -5) {$R_1$};
 \node (R_2) at (8, -5) {$R_2$};
 \node (P_01) at (3,-1.7) {};
 \node (P_12) at (5,-1.7) {};
 \node (R_01) at (5,-4.2) {};
 \node (R_12) at (7,-4.2) {};
 \node (S_12) at (3.2,0.7) {};
  \draw
    (S_0) edge[->,thick,bend left] node [above] {$E$} (S_1)
       (S_1) edge[->,thick,bend left] node [below] {$F_1$} (S_0)
    (P_0) edge[->,thick,bend left] node [above] {$S_1$} (P_1)
    (P_1) edge[->,thick,bend left] node [above] {$S_1$} (P_2)
    (P_1) edge[->,thick,bend left] node [below] {$F_2$} (P_0)
    (P_2) edge[->,thick,bend left] node [below] {$F_2$} (P_1)
    (R_0) edge[->,thick,bend left] node [above] {$P_2$} (R_1)
    (R_1) edge[->,thick,bend left] node [above] {$P_2$} (R_2)
    (R_1) edge[->,thick,bend left] node [below] {$F_3$} (R_0)
    (R_2) edge[->,thick,bend left] node [below] {$F_3$} (R_1)
     (S_1) edge[->,thick,bend left] node [above] {$R_2$} (S_2)
       (S_2) edge[->,thick,bend left] node [below] {$G$} (S_1);
\draw[dashed, ->,thick,looseness=1] (S_1) to [out=-90,in=90] (P_01) ;    
\draw[dashed, ->,thick,looseness=1] (S_1) to [out=-90,in=90] (P_12) ; 
\draw[dashed, ->,thick,looseness=1] (P_2) to [out=-90,in=90] (R_01) ;    
\draw[dashed, ->,thick,looseness=1] (P_2) to [out=-90,in=90] (R_12) ;  
\path[dashed,  ->,thick] (R_2) edge [out=60, in=0] (S_12);
\end{tikzpicture}
}
\caption{Three-layer ERK cascade with feedback loop}
\label{fig:3layerloop}
\end{figure}
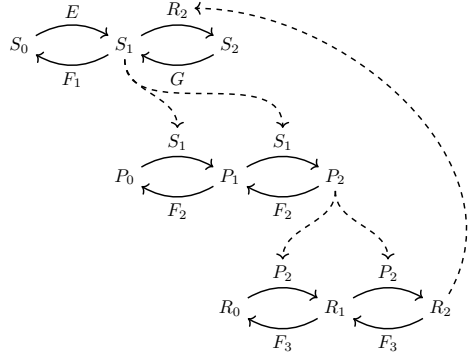
We can compute the circuits of multistationarity of variations obtained using equal or different phosphatases 
in the three layers of the cascade. We summarize the results obtained in Table  \ref{table:3layerwithloop}.

\begin{table}[h!]
\centering
\begin{small}
\begin{tabular}{c|c}
Enzymes & Canonical Circuits of Multistationarity \\
\hline
\multirow{2}{*}{$F_1,F_1,F_1$}
& $ \{P_0+S_1\}, \{P_1+F_1\},
    \{P_2+F_1\}, \{R_0+P_2\}, \{R_1+F_1\}, 
    \{R_2+F_1\},$ \\
& $ \{P_1+S_1, S_1+F_1\} $ \\
\hline
\multirow{2}{*}{$F_1,F_1,F_3$}
& $  \{P_0+S_1\}, \{P_1+F_1\},
  \{P_2+F_1\}, 
  \{R_0+P_2\},  \{R_2+F_3\}, 
 $\\ 
& $ \{P_1+S_1, S_1+F_1\}$ \\  
\hline
\multirow{2}{*}{$F_1,F_2,F_1$}
& $ \{P_0+S_1\}, \{P_2+F_2\}, 
  \{R_0+P_2\},   \{R_1+F_1\},  \{R_2+F_1\}, 
  $ \\
& $
 \{P_1+S_1, R_1+P_2, S_1+F_1\}, 
 \{P_1+F_2, P_1+S_1\}$ \\
\hline
\multirow{2}{*}{$F_1,F_2,F_2$}
& 
$ \{P_0+S_1\}, \{P_2+F_2\},
   \{R_0+P_2\}, \{R_1+F_2\}, \{R_2+F_2\}, $\\ 
& $
    \{P_1+F_2,R_1+P_2\},
    \{P_1+F_2,P_1+S_1\}
   $ \\
\hline
\multirow{2}{*}{$F_1,F_2,F_3$}
& 
$ \{P_0+S_1\}, \{P_2+F_2\}, 
  \{R_0+P_2\}, \{R_2+F_3\},  $\\ 
& $  
    \{P_1+F_2,P_1+S_1\},
\{ R_1+F_3,R_1+P_2\} $ \\     

\vspace{0.3cm} 
\end{tabular}
\end{small}
\caption{Canonical circuits of multistationarity for the three-layer  ERK cascade with feedback loop. \label{table:3layerwithloop}}
\end{table}
\end{example}

A closer look at Tables~\ref{table:3layernoloop} and \ref{table:3layerwithloop} reveals that all circuits from Table~\ref{table:3layernoloop} are present in Table~\ref{table:3layerwithloop}. 
The only new ones in Table~\ref{table:3layerwithloop} are $\{P_1+F_2, P_1+S_1\}$ and $\{R_1+F_3, R_1+P_2\}$. 

In all cases, the one element circuits $ \{P_0+S_1\}, \{P_2+F_i\}, \{R_0+P_2\}
, \{R_1+F_j\}, \{R_2+F_j\}$ are present and in the case in which the last and first phosphatases are the same but different from the second
one, then $\{P_1+S_1, R_1+P_2, S_1+F_1\}$  is a circuit of multistationarity of cardinality $3$ involving reactions in each of
the three layers. It would be certainly interesting to include the biochemical viewpoint about the occurrence of multistationarity.

\section{Hybrid phopho/dephosphorylation mechanisms}\label{sec:sequential}

We now consider the $n$-site phosphorylation system. For a more complete discussion of this topic see~\cite{Gu2007} and~\cite{SK2015}. In~\cite{GRPMD} the authors studied the $n$-site sequential distributive phosphorylation, see Figure~\ref{fig:nsite}~(A).
In this kind of networks, there is an initial substrate $S_0$  that undergoes a post-translational modification by the attachment of  phosphate groups to different binding sites.  In a sequential phosphorylation/dephosphorylation mechanism, sites are modified in a specific order. In this case, we denote by $S_i$ the phospho-form with $i$ phosphorylated sites. In a distributive mechanism, sites are phosphorylated/dephosphorylated one by one. If all the sites are phosphorylated/dephosphorylated at once  it is said that the enzymes act processively. In what follows we discuss the case when the enzyme $E$ acts distributively and the enzyme $F$ acts in a mixed fashion, part distributively and part processively, see Figure~\ref{fig:nsite}~(B).

\begin{figure}[h!] 
\scalebox{0.7}{
\begin{tabular}{cc}
 (A) & \begin{minipage}{0.85\textwidth}\begin{tikzpicture}[xscale=1,yscale=1]
 \node (S_0) at (0, 0) {$S_0$};
 \node (S_1) at (2, 0) {$S_1$};
 \node (S_2) at (4, 0) {$\cdots$};
 \node (S_{n-1}) at (6, 0) {$S_{n-1}$};
 \node (S_n) at (8, 0) {$S_{n}$};
 \draw (S_0) edge[->,thick,bend left] node [above] {$E$} (S_1)
       (S_1) edge[->,thick,bend left] node [below] {$F$} (S_0)
       (S_1) edge[->,thick,bend left] node [above] {$E$} (S_2)
       (S_2) edge[->,thick,bend left] node [below] {$F$} (S_1)
       (S_2) edge[->,thick,bend left] node [above] {$E$} (S_{n-1})
       (S_{n-1}) edge[->,thick,bend left] node [below] {$F$} (S_2)
       (S_{n-1}) edge[->,thick,bend left] node [above] {$E$} (S_n)
       (S_n) edge[->,thick,bend left] node [below] {$F$} (S_{n-1});
 \end{tikzpicture} \end{minipage}\\
(B) & \begin{minipage}{0.85\textwidth}\begin{tikzpicture}[xscale=1,yscale=1]
 \node (S_{i_0}) at (0, 0) {$S_{0}$};
 \node (P1) at (2,0) {$\cdots$};
 \node (S_{i_1}) at (4, 0) {$S_{i_1}$}; 
 \node (P2) at (6,0) {$\cdots$};
 \node (S_{i_{k-1}}) at (8, 0) {$S_{i_{k-1}}$};
 \node (P3) at (10,0) {$\cdots$};
 \node (S_{i_k}) at (12, 0) {$S_n$};
 \draw (S_{i_0}) edge[->,thick,bend left] node [above] {$E$} (P1)
       (P1) edge[->,thick,bend left] node [above] {$E$} (S_{i_1})
       (S_{i_1}) edge[->,thick,bend left] node [above] {$E$} (P2)
       (P2) edge[->,thick,bend left] node [above] {$E$} (S_{i_{k-1}})
       (S_{i_{k-1}}) edge[->,thick,bend left] node [above] {$E$} (P3)
       (P3) edge[->,thick,bend left] node [above] {$E$} (S_{i_k})
       (S_{i_1}) edge[->,thick,bend left] node [below] {$F$} (S_{i_0})
       (S_{i_{k-1}}) edge[->,thick,bend left] node [below] {$F$} (P2)  
       (P2) edge[->,thick,bend left] node [below] {$F$} (S_{i_1})    
       (S_{i_k}) edge[->,thick,bend left] node [below] {$F$} (S_{i_{k-1}});
 \end{tikzpicture}\end{minipage}
\end{tabular}
}
\caption{(A): $n$-site sequential distributive phosphorylation. (B): $n$-site sequential but not distributive phosphorylation.}
\label{fig:nsite}
\end{figure}
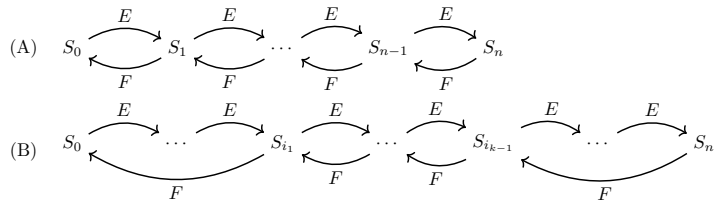

  
\begin{definition}\label{def:CI}
 Let $n$ be a positive integer and fix a set of indices $I=\{i_0=0<i_1<\dots <i_k=n\}$. The network $C_{I}$ is defined by species $S_0,S_1,\dots,S_n,E,F$ and the following reactions:
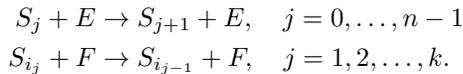
\begin{align*}
S_j+E&\to S_{j+1}+E,\quad j=0,\dots, n-1 \\
S_{i_j}+F&\to S_{i_{j-1}}+F,\quad j=1,2,\dots,k.
\end{align*}
\end{definition}

We computed the circuits of multistationarity of the network $C_{I}$ for all possible $I$ for  up to $n=8$ binding sites and our experiments suggested Theorem~\ref{th:nsitepartial} below, valid for any value of $n$. 

As the networks $C_I$ are MESSI networks (see Section~\ref{sec:Messi}), it follows from Theorems~\ref{th:ylebn} and~\ref{th:monoy} that all  $C_I$ are monostationary and lebn and that any non-confluent extension of them is also lebn. Moreover, they are {$s$-toric MESSI systems}, defined in~\cite{messi}. 

We now restate with our language Theorem~5.4 in~\cite{messi}, which is essentially contained in~\cite{Focm}.
Given a lebn network, with associated matrix  $\rm {Exp}$ as in Definition~\ref{def:Bpoly} of full rank, we will denote by $A \in \Z^{d \times s}$ a full rank
Gale dual matrix of Exp. This means that $A \cdot \rm{Exp} = 0$.

\begin{theorem}[Restatement of Th.~5.4 in~\cite{messi}] \label{th:5.3}
Let $G$ be a lebn network, with conservation-law matrix $W$ and associated matrix  $\rm {Exp}$ as in Definition~\ref{def:Bpoly}. 
Assume that ${\rm rank}(W) + {\rm rank}({\rm Exp})=s$ and~\eqref{eq:dets} holds. Then, the following assertions are equivalent:
\begin{itemize}
 \item[(i)] $G$ is monostationary
 \item[(ii)] For all $I \subset  \{1, \dots, s\}$ of cardinality ${\rm rank}(W)$,
  \[{\rm sign}(\det(W_I)) = \varepsilon  \, {\rm sign} (\det(A_{I})),\] 
 with $\varepsilon = \pm 1$.
 \item[(iii)] All the coefficients in the polynomial $B$ in~\eqref{eq:Bpoly} have the same sign. 
\end{itemize}

\end{theorem}

We can now prove the main result in this section. Note that the case where $I=\{0,1,\dots,n\}$ is Proposition~4.11 in \cite{SFeliu}.

\begin{theorem}\label{th:nsitepartial} Let $n$ be a positive integer and fix a set $I=\{i_0=0<i_1<\dots <i_k=n\}$. Consider the network $C_{I}$, as in Definition~\ref{def:CI}, the circuits of multistationarity of $C_{I}$ are the sets defined by  every single source complex but $S_{i_{k-1}}+E, S_{i_{k-1}+1}+E, \dots, S_{n-1}+E$ 
and $S_{i_1}+F$.
\begin{center}
 \begin{tikzpicture}[xscale=1,yscale=1]
 \node (S_{i_0}) at (0, 0) {$S_{0}$};
 \node (P1) at (2,0) {$\cdots$};
 \node (S_{i_1}) at (4, 0) {$S_{i_1}$}; 
 \node (P2) at (6,0) {$\cdots$};
 \node (S_{i_{k-1}}) at (8, 0) {$S_{i_{k-1}}$};
 \node (P3) at (10,0) {$\cdots$};
 \node (S_{i_k}) at (12, 0) {$S_{n}$};
 \draw (S_{i_0}) edge[->,thick,bend left] node [above] {$E$} (P1)
       (P1) edge[->,thick,bend left] node [above] {$E$} (S_{i_1})
       (S_{i_1}) edge[->,thick,bend left] node [above] {$E$} (P2)
       (P2) edge[->,thick,bend left] node [above] {$E$} (S_{i_{k-1}})
       (S_{i_{k-1}}) edge[->,thick,bend left,dashed] node [above] {$E$} (P3)
       (P3) edge[->,thick,bend left,dashed] node [above] {$E$} (S_{i_k})
       (S_{i_1}) edge[->,thick,bend left,dashed] node [below] {$F$} (S_{i_0})
       (S_{i_{k-1}}) edge[->,thick,bend left] node [below] {$F$} (P2)  
       (P2) edge[->,thick,bend left] node [below] {$F$} (S_{i_1})    
       (S_{i_k}) edge[->,thick,bend left] node [below] {$F$} (S_{i_{k-1}});
 \end{tikzpicture}
\end{center}
If $i_1=n$, i.e. if $I=\{0,n\}$, the addition of intermediates cannot produce multistationarity. 
\end{theorem}

In particular, if $I=\{i_0=0<i_1=m<i_2=n\}$ for some $0<m<n$ then the circuits of multistationarity are 
$$\{S_0+E\},\{S_1+E\},\dots,\{S_{m-1}+E\}, \{S_{n}+F\}.$$

In the case $I=\{i_0=0<i_1=a<i_2=b<i_3=n\}$, the circuits of multistationarity are
$$\{S_0+E\},\{S_1+E\},\dots,\{S_{b-1}+E\}, \{S_{b}+F\}, \{S_{n}+F\}.$$

\begin{proof} We apply Theorem~\ref{th:5.3}. Suppose that we add one intermediate $U_{\ell_0}$ in the canonical form: 
\begin{equation*}
 {{S_{\ell_0}}}+{E}
 \rightleftarrows U_{\ell_0},
  \text{ for certain } \ell_0 \text{ with } i_{j_0}\leq \ell_0 < i_{j_0+1},\, \text{and } 0\leq {j_0}<k-1. 
 \end{equation*} 
By Theorem~\ref{th:ylebn}, this new network is lebn and we can therefore parametrize the positive steady states by monomials. Call with small letters the concentrations of the species. We can write the concentration of all the species at steady state in terms of the concentrations $s_0$, $e$, and $f$ in the following way:
\begin{equation}\label{parametrization}
s_{\ell}=\psi_{\ell}(\kappa)\dfrac{s_0 e^j}{f^j}\ \text{ if } i_j\leq \ell < i_{j+1},
\ \text{ and }\ u_{\ell_0}=\phi(\kappa)\dfrac{s_0 e^{j_0+1}}{f^{j_0}},
\end{equation}
where $\psi_{\ell}(\kappa)$ and $\phi(\kappa)$ are rational functions of the rate constants $\kappa$ which
are defined in the positive orthant.

We consider the matrix $A$ of exponents of the monomials in the parametrization~\eqref{parametrization} with the following order of the species: $S_0=S_{i_0}$, $S_1$, $\dots$, $S_n=S_{i_k}$, $E$, $F$ and $U_{\ell_0}$. 
\begin{small}
\[A= \bordermatrix{%
     &  S_{i_0} & & S_{i_1 - 1}       &  S_{i_1} & S_{i_1 + 1}  & & S_{i_k -1 } & S_{i_k} & E & F & Y_{\ell_0} \cr
     &  1 & \cdots &1 & 1 & 1 &\cdots & 1  & 1 & 0 & 0 & 1\cr
     & 0 & \cdots   & 0 & 1 & 1 & \cdots &  k-1 & k & 1 & 0 & j_0+1\cr
     & 0 & \cdots & 0 & -1 & -1 & \cdots &  -(k-1) & -k & 0 & 1 & -j_0
  }.
\]
\end{small}%
We also consider the matrix $W$ of conservation relations:
\begin{small}
\[W= \bordermatrix{%
     &  S_{i_0} & & S_{i_1 - 1}       &  S_{i_1} & S_{i_1 + 1}  & & S_{i_k -1 } & S_{i_k} & E & F & Y_{\ell_0} \cr
     &  1 & \cdots &1 & 1 & 1 &\cdots & 1  & 1 & 0 & 0 & 1\cr
     & 0 & \cdots   & 0 & 0 & 0 & \cdots &  0 & 0 & 1 & 0 & 1\cr
     & 0 & \cdots & 0 & 0 & 0 & \cdots &  0 & 0 & 0 & 1 & 0
  }.
\]
\end{small}%
If we take the submatrices of $W$ and $A$ with columns corresponding to the species $S_0=S_{i_0}$, $E$ and $F$, the corresponding determinants have the same sign. 
But if we take the submatrices of $W$ and $A$ with columns corresponding to the species $S_n=S_{i_k}$, $F$ and $Y_{\ell_0}$, the determinants have opposite sign, and then the network is multistationary. 

When we add one intermediate $U_j$ in the canonical form
  \begin{equation*}
 {S_{i_{j}}}+{F} \rightleftarrows U_j, \text{ with } 2\leq j \leq k, 
\end{equation*}
 the network is also multistationary and the proof is analogous to the previous one. In this case, we can obtain a parametrization of the concentration of the species at steady state similar as in~\eqref{parametrization} in terms of $s_0$, $e$ and $f$. We have the corresponding matrix of monomials:
\begin{small}
\[A= \bordermatrix{%
     &  S_{i_0} & & S_{i_1 - 1}       &  S_{i_1} & S_{i_1 + 1}  & & S_{i_k -1 } & S_{i_k} & E & F & U_{j} \cr
     &  1 & \cdots &1 & 1 & 1 &\cdots & 1  & 1 & 0 & 0 & 1\cr
     & 0 & \cdots   & 0 & 1 & 1 & \cdots &  k-1 & k & 1 & 0 & j\cr
     & 0 & \cdots & 0 & -1 & -1 & \cdots &  -(k-1) & -k & 0 & 1 & -(j-1)
  },
\]
\end{small}%
and the matrix of conservation relations:
\begin{small}
\[W= \bordermatrix{%
     &  S_{i_0} & & S_{i_1 - 1}       &  S_{i_1} & S_{i_1 + 1}  & & S_{i_k -1 } & S_{i_k} & E & F & U_{j} \cr
     &  1 & \cdots &1 & 1 & 1 &\cdots & 1  & 1 & 0 & 0 & 1\cr
     & 0 & \cdots   & 0 & 0 & 0 & \cdots &  0 & 0 & 1 & 0 & 0\cr
     & 0 & \cdots & 0 & 0 & 0 & \cdots &  0 & 0 & 0 & 1 & 1
  }.
\]
\end{small}%
As before, if we take the submatrices of $W$ and $A$ with columns corresponding to the species $S_0=S_{i_0}$, $E$ and $F$, the corresponding determinants have the same sign. But if we take the submatrices of $W$ and $A$ with columns corresponding  to the species $S_0=S_{i_{i_0}}$, $E$ and $U_{j}$, the corresponding determinants have opposite sign, and then the network is multistationary.

Now, it remains to check that if we add the all the intermediates $U_{\ell}$ for each $i_{k-1}\leq \ell< i_{k}=n$ and the intermediate $U_{i_1}$ in the canonical form, the network is monostationary. That is, if we add simultaneously  all the intermediates in the following form:
\begin{equation*}
 {{S_{\ell}}}+{E} \rightleftarrows Y_{\ell},
  \text{ for all } \ell \text{ with } i_{k-1}\leq \ell< i_{k},   \qquad
 {S_{i_{1}}}+{F} \rightleftarrows U_{i_1},
 \end{equation*}
the network is monostationary. In this case, any of its subnetworks is also monostationary and we are done.  We can also express all the concentrations of the species in terms of $s_0$, $e$, and $f$, in a similar way as in parametrization~\eqref{parametrization}. Then, we take the corresponding matrices $A$ and $W$ with the following order of species:  $S_0=S_{i_0}$, $S_1$, $\dots$, $S_n=S_{i_k}$, $E$, $F$ and $Y_{i_{k-1}}$,$\dots$,$Y_{i_k-1}$ and $U_{i_1}$: 
 
\begin{small}
\[A= \bordermatrix{%
     &  S_{i_0} &  & S_{i_k} & E & F & Y_{i_{k-1}} & & Y_{i_{k}-1} & U_{i_1} \cr
     &  1 & \cdots &1 & 0 & 0 & 1 & \cdots & 1 & 1 \cr
     & 0 & \cdots & k & 1 & 0 & k & \cdots & k & 1 \cr
     & 0 & \cdots & -k & 0 & 1 & -(k-1) & \cdots & -(k-1)  & 0
  },
\]
\end{small}%

\begin{small}
\[W= \bordermatrix{%
     &  S_{i_0} &  & S_{i_k} & E & F & Y_{i_{k-1}} & & Y_{i_{k}-1} & U_{i_1} \cr
     &  1 & \cdots  & 1 & 0 & 0 & 1 &  \cdots  & 1 & 1  \cr
     & 0 & \cdots   & 0 & 1 & 0 & 1 & \cdots   & 1 & 0 \cr
     & 0 & \cdots  & 0 & 0 & 1 & 0 & \cdots  & 0 & 1
  }.
\] 
\end{small}%

Note that the first columns corresponding to the substrates $S_{\ell}$ are the same as in the matrices of the previous cases. Once again, if we take the submatrices of $W$ and $A$ with columns corresponding to the species $S_0=S_{i_0}$, $E$ and $F$, the corresponding determinants are equal to one. Now we have to check that  $\det(W_J)\det(A_J)\geq 0,$ for all subset $J$, with $|J|=3$ and we can conclude that the network is monostationary.
 
Since $A$ and $W$ have several repeated columns it is enough to check that for all integers $j$ with $0\leq j \leq k$ and for all subsets $J$ with $|J|=3$ we have $\det(W'_J)\det(A'_J)\geq 0,$ where
\begin{small}
\[
A'= \begin{pmatrix}
1&0&0&1&1\\
j&1&0&k&1\\
-j&0&1& -(k-1) & 0
\end{pmatrix}
\quad \text{ and } \quad
W'= \begin{pmatrix}%
1&0&0&1&1\\
0&1&0&1&0\\
0&0&1&0&1
\end{pmatrix}.
\] 
\end{small}%
This can be easily done in a Computer Algebra System or even by hand.

\end{proof}

\section{Certified lebn biochemical mechanisms} \label{sec:Messi}

  In this section we show that many common examples of MESSI systems can be easily checked to be lebn.
The key feature of MESSI networks is that the chemical species are grouped into different subsets according to the way they participate in the reactions, very much akin to the intuitive partition of the species according to their function. 
We briefly introduce the basic definitions. For a more detailed explanation, see \cite{messi}. 
  We give in Theorems~\ref{th:ylebn} and~\ref{th:monoy} simple combinatorial conditions that ensure that a  core network is monostationary or linearly equivalent to a binomial network.  Thus, any  MESSI system without intermediate species that satisfies the hypotheses in Theorem~\ref{th:ylebn} is lebn  and we can take $C=\mathscr{C}_G$ to obtain all the multistationarity circuits, as in~\cite{SFeliu}. In particular, we show in Example~\ref{ex:3-layer-again} that this happens for the three-layer cascade without intermediate species studied in
 Section~\ref{sec:ERK}.

\medskip

\begin{definition} \label{def:messi}
A MESSI structure on a network with set of species $\Sp$ is given by a partition as in~\eqref{eq:partition}, together with a partition on the set of complexes and restrictions on the possible reactions described below:
\begin{equation}\label{eq:partition}
\Sp=\mathcal{I}\, \bigsqcup \Sp^{(1)} \bigsqcup \Sp^{(2)} \bigsqcup \dots \bigsqcup \Sp^{(d)}, \quad d \ge 1,
\end{equation}
where $\bigsqcup$ denotes disjoint union. Species in $\mathcal{I}$ are intermediate species and the species in $\Sp_M:= \Sp \setminus\mathcal{I}$ are core species. 

Complexes are also partitioned into intermediate complexes and core complexes, as we have defined in \S~\ref{sec:intermediates}. However, core complexes in MESSI networks must satisfy the following two conditions:
\begin{enumerate}[(i)]
\item They are monomolecular or bimolecular and consist of either one or two core species.  
\item If the core complex consists of two species $X_i, X_j$, they {\em must} belong to \emph{different} sets 
$\Sp^{(\alpha)}, \Sp^{(\beta)}$ (with $\alpha \neq \beta$ and $\alpha, \beta \geq1$).
\end{enumerate}

\medskip

The reactions of MESSI networks are constrained by the following rules:
\begin{itemize}
\item[(iii)] If three species are related by $X_i+ X_j \uri X_k$ or $X_k \uri X_i + X_j$, then $X_k\in\mathcal{I}$.
 \item[(iv)] If two core species $X_i, X_j$ are related by $X_i\uri X_j$, then there exists $\alpha \ge 1$ such that both belong to $\Sp^{(\alpha)}$.
 \item[(v)] If $X_i+X_j\uri X_k+X_\ell$, then there exist $\alpha \neq \beta$ such that $X_i,X_k \in 
 \Sp^{(\alpha)}$, $X_j,X_\ell \in \Sp^{(\beta)}$ or $X_i,X_\ell \in \Sp^{(\alpha)}$, $X_j,X_k \in \Sp^{(\beta)}$.
\end{itemize}

A {\em MESSI system} is the mass-action kinetics dynamical system~\eqref{eq:CRN} associated with a MESSI network. 
\end{definition}

\begin{remark}
 Note that constraint~(iii) imposed on the reactions makes the definition of intermediate species and complexes in MESSI networks slightly different from the definitions of core and intermediate species in \S~\ref{sec:intermediates}. Any intermediate species of a MESSI network can be considered an intermediate species but some core species as in~\S~\ref{sec:intermediates} could only be intermediate species in the MESSI setting. For example, if we consider the network
\[
 X_1+X_2\longrightarrow X_3 \longrightarrow X_4 \longrightarrow X_1+X_5,
\]
then necessarily $\{X_3,X_4\}\subseteq \mathcal{I}$.
\end{remark}

Given two partitions, $\Sp$ and $\Sp'$ we say that $\Sp$ refines $\Sp'$ if and only if the set of intermediate species of $\Sp$ contains the set of intermediate species of $\Sp'$ and for every set of core species of $\Sp$, $\Sp^{(\alpha)}$, there exists a set of core species of $\Sp'$, $\Sp'^{(\beta)}$, such that $\Sp^{(\alpha)}\subseteq \Sp'^{(\beta)}$. With this partial order we have the notion of minimal partition.

We now present three associated digraphs $G_1$, $G_2^0$, and $G_E$ associated to a MESSI network $G$. We refer the reader to \cite{messi} for complete definitions.

\begin{definition}
 \label{def:digraphs}
Given a MESSI network $G$, we call $G_1$ the associated digraph of the reduced network $G_{red, \mathcal I}$ introduced in Definition~\ref{def:redext}. The labels assigned to the edges are the rational functions of the original rate constants $\kappa$ defined in~\eqref{eq:tau}.
In order to construct the digraph $G_2^0$ we first ``hide'' the concentrations of some of the species in the labels.
We keep all monomolecular reactions $X_i\to X_j$ and for each reaction $X_i+X_\ell \overset{\tau}{\longrightarrow} X_j+X_m$, with $X_i,X_j \in \Sp^{(\alpha)}$, $X_\ell,X_m \in \Sp^{(\beta)}$, we consider two reactions $X_i \overset{\tau x_\ell}{\longrightarrow} X_j$ and $X_\ell \overset{\tau x_i}{\longrightarrow} X_m$.
We obtain a multidigraph $MG_2$ that may contain loops or parallel edges between some pairs of nodes (i.e., directed edges with the same source and target nodes). We define the digraph $G_2^0$ by deleting loops and isolated nodes and by collapsing into one edge all parallel edges in $MG_2$. We define the labels of each edge as the sum of the labels of the corresponding collapsed edges in $MG_2$. Note that these labels might depend on some of the concentrations.
We finally define the associated digraph $G_E$. The set of vertices of $G_E$ equals $\{\Sp^{(\alpha)}\mid \, \alpha \ge 1\}$. The pair $(\Sp^{(\alpha)},\Sp^{(\beta)})$ is an edge of $G_E$ when there is a species in $\Sp^{(\alpha)}$ in a label of an edge in $G_2^0$ between (distinct) species of $\Sp^{(\beta)}$. 
\end{definition}

\begin{remark}\label{rmk:eq_from_G2^0}
 Note that $G_1$ and $G_2^0$ together with the equations of the intermediate species define the whole variety of steady states of the system associated to $G$. Moreover, when the set of intermediate species is empty ($\mathcal{I}=\emptyset$), then the equations of $G$ can be reconstructed from $G_2^0$.
\end{remark}

\begin{example}[Three-layer cascade, continued] \label{ex:3-layer-again}
Recall the three-layer cascade network in Example~\ref{ex:3layercascade}. 
The $s=10$ species of the network are:
\begin{center}
\begin{tabular}{lllllll}
$X_1$=\ce{S_0},& $X_3$=\ce{P_0},& $X_5$=\ce{P_2},&$X_7$=\ce{R_1}, & $X_9$=\ce{E}, &  $X_{11}$=\ce{F_2}~.\\
$X_2$=\ce{S_1},& $X_4$=\ce{P_1},& $X_6$=\ce{R_0},& $X_8$=\ce{R_2},& $X_{10}$=\ce{F_1},\\
\end{tabular}
\end{center}
We consider the following partition: 
$$\mathcal{I}=\{ES_0,F_1S_1,S_1P_0, S_1P_1, F_2P_2,F_2P_1,P_2R_0, P_2R_1, F_2R_2, F_2R_1\},$$  
is the set of intermediate species, and
$$\Sp^{(1)}=\{S_0,S_1\}, \Sp^{(2)}=\{P_0,P_1,P_2\}, \Sp^{(3)}=\{R_0,R_1,R_2\},$$
$$\Sp^{(4)}=\{E\}, \Sp^{(5)}=\{F_1\}, \Sp^{(6)}=\{F_2\},$$ 
are the subsets of core species.
The intermediate complexes correspond to the intermediate species, and the remaining complexes are core complexes. 
This partition defines a MESSI structure in the network. Moreover, this is the only possible minimal partition (up to order) for the (core) species.

We present the digraphs $G_1$, $G_2^\circ$, $G_E$, and the multidigraph $MG_2$ associated to the three-layer cascade network in Figure~\ref{fig:3-layerMESSI}.
\begin{figure}
\scalebox{0.7}{
\begin{tabular}{l@{}l}
 $G_1$: &
 \begin{tabular}[t]{rl}
\ce{S_{0} + E ->[\tau_1] S_{1} + E},  &
\ce{S_{1} + F_1 ->[\tau_2] S_{0} + F_1}\\
\ce{P_{0} + S_{1} ->[\tau_3] P_{1} + S_{1}  ->[\tau_4] P_{2} + S_{1} }, &
\ce{P_{2} + F_2 ->[\tau_5] P_{1} + F_2 ->[\tau_6] P_{0} + F_2} \\
\ce{R_{0} + P_{2} ->[\tau_7] R_{1} + P_{2}  ->[\tau_8] R_{2} + P_{2} }, &
\ce{R_{2} + F_2 ->[\tau_9] R_{1} + F_2 ->[\tau_{10}] R_{0} + F_2}. \\
\end{tabular}\\
\multirow{2}{*}{$MG_2$:} & \multirow{12}{*}{
 \begin{tikzpicture}[node distance=1.7cm]
  \node (S0) {\ce{S_{0}}};
  \node (S1) [right of=S0] {\ce{S_{1}}};
  \node (P0) [below of=S0] {\ce{P_{0}}};
  \node (P1) [right of=P0] {\ce{P_{1}}};
  \node (P2) [right of=P1] {\ce{P_{2}}};
  \node (R0) [below of=P0] {\ce{R_{0}}};
  \node (R1) [right of=R0] {\ce{R_{1}}};
  \node (R2) [right of=R1] {\ce{R_{2}}};
  \draw ($(S0.east) + (0,.2em)$) edge[->] node[above]{$\tau_1 x_9$} ($(S1.west) + (0,.2em)$);
  \draw ($(S1.west) + (0,-.2em)$) edge[->] node[below]{$\tau_2 x_{10}$} ($(S0.east) + (0,-.2em)$);
  \draw ($(P0.east) + (0,.2em)$) edge[->] node[above]{$\tau_3 x_2$} ($(P1.west) + (0,.2em)$);
  \draw ($(P1.west) + (0,-.2em)$) edge[->] node[below]{$\tau_6 x_{11}$} ($(P0.east) + (0,-.2em)$);
  \draw ($(P1.east) + (0,.2em)$) edge[->] node[above]{$\tau_4 x_2$} ($(P2.west) + (0,.2em)$);
  \draw ($(P2.west) + (0,-.2em)$) edge[->] node[below]{$\tau_5 x_{11}$} ($(P1.east) + (0,-.2em)$);
  \draw ($(R0.east) + (0,.2em)$) edge[->] node[above]{$\tau_7 x_5$} ($(R1.west) + (0,.2em)$);
  \draw ($(R1.west) + (0,-.2em)$) edge[->] node[below]{$\tau_{10} x_{11}$} ($(R0.east) + (0,-.2em)$);
  \draw ($(R1.east) + (0,.2em)$) edge[->] node[above]{$\tau_8 x_5$} ($(R2.west) + (0,.2em)$);
  \draw ($(R2.west) + (0,-.2em)$) edge[->] node[below]{$\tau_9 x_{11}$} ($(R1.east) + (0,-.2em)$);
  \path[]
    (S1) edge [loop right] node {$\tau_1 x_1$} (S1)
         edge [->,loop right,in=-30,out=30,looseness=15] node[above=8pt] {$\tau_4 x_4$} (S1)
    (P2) edge [loop right] node {$\tau_7 x_6$} (P2)
         edge [->,loop right,in=-30,out=30,looseness=15] node[above=8pt] {$\tau_8 x_7$} (P2);
  \node (E) [right=6.5cm of S0] {\ce{E}};
  \node (F1) [below of=E] {\ce{F_1}};
  \node (F2) [below of=F1] {\ce{F_2}};
  \path[]
    (E) edge [loop right] node {$\tau_3 x_3$} (E)
    (F1) edge [loop right] node {$\tau_2 x_2$} (F1)
    (F2) edge [loop right] node {$\tau_5 x_5$} (F2)
         edge [->,loop right,in=-20,out=20,looseness=20] node {$\tau_6 x_4$} (F2)
         edge [->,loop right,in=-20,out=20,looseness=35] node {$\tau_9 x_8$} (F2)
         edge [->,loop right,in=-20,out=20,looseness=50] node {$\tau_{10} x_7$} (F2);    
\end{tikzpicture}
} \\
 & \\ 
 & \\
 & \\
 & \\
 & \\
 & \\
 & \\
 & \\
 & \\
 & \\
 & \\
 & \\
 & \\
\multicolumn{2}{l}{
\begin{tabular}{llll}
   $G_2^\circ$: & 
\multirow{10}{*}{
 \begin{tikzpicture}[node distance=1.7cm]
  \node (S0) {\ce{S_{0}}};
  \node (S1) [right of=S0] {\ce{S_{1}}};
  \node (P0) [below of=S0] {\ce{P_{0}}};
  \node (P1) [right of=P0] {\ce{P_{1}}};
  \node (P2) [right of=P1] {\ce{P_{2}}};
  \node (R0) [below of=P0] {\ce{R_{0}}};
  \node (R1) [right of=R0] {\ce{R_{1}}};
  \node (R2) [right of=R1] {\ce{R_{2}}};
  \draw ($(S0.east) + (0,.2em)$) edge[->] node[above]{$\tau_1 x_9$} ($(S1.west) + (0,.2em)$);
  \draw ($(S1.west) + (0,-.2em)$) edge[->] node[below]{$\tau_2 x_{10}$} ($(S0.east) + (0,-.2em)$);
  \draw ($(P0.east) + (0,.2em)$) edge[->] node[above]{$\tau_3 x_2$} ($(P1.west) + (0,.2em)$);
  \draw ($(P1.west) + (0,-.2em)$) edge[->] node[below]{$\tau_6 x_{11}$} ($(P0.east) + (0,-.2em)$);
  \draw ($(P1.east) + (0,.2em)$) edge[->] node[above]{$\tau_4 x_2$} ($(P2.west) + (0,.2em)$);
  \draw ($(P2.west) + (0,-.2em)$) edge[->] node[below]{$\tau_5 x_{11}$} ($(P1.east) + (0,-.2em)$);
  \draw ($(R0.east) + (0,.2em)$) edge[->] node[above]{$\tau_7 x_5$} ($(R1.west) + (0,.2em)$);
  \draw ($(R1.west) + (0,-.2em)$) edge[->] node[below]{$\tau_{10} x_{11}$} ($(R0.east) + (0,-.2em)$);
  \draw ($(R1.east) + (0,.2em)$) edge[->] node[above]{$\tau_8 x_5$} ($(R2.west) + (0,.2em)$);
  \draw ($(R2.west) + (0,-.2em)$) edge[->] node[below]{$\tau_9 x_{11}$} ($(R1.east) + (0,-.2em)$); 
\end{tikzpicture} 
} & 
$G_E:$ & 
\multirow{10}{*}{
\begin{tikzpicture}[ampersand replacement=\&] 
\matrix (m) [matrix of math nodes, row sep=1.5em, column sep=1em, text height=1.5ex, text depth=0.25ex]
{ \Sp^{(1)}\&  \Sp^{(2)} \& \Sp^{(3)} \\
\Sp^{(4)} \&  \Sp^{(5)} \& \Sp^{(6)} \\};
\draw[->]($(m-1-1)+(0.4,0)$) to node[below] (x) {}($(m-1-2)+(-0.4,0)$) ;
\draw[->]($(m-1-2)+(0.4,0)$) to node[below] (x) {} ($(m-1-3)+(-0.4,0)$);
\draw[->]($(m-2-3)+(0,0.3)$) to node[below] (x) {} ($(m-1-3)+(0,-0.2)$);
\draw[->]($(m-2-3)+(0,0.3)$) to node[below] (x) {} ($(m-1-2)+(0.35,-0.15)$);
\draw[->]($(m-2-2)+(0,0.3)$) to node[below] (x) {} ($(m-1-1)+(0.35,-0.15)$);
\draw[->]($(m-2-1)+(0,0.3)$) to node[below] (x) {} ($(m-1-1)+(0,-0.2)$);
\end{tikzpicture}
} \\
 & \\
 & \\
 & \\
 & \\
 & \\
 & \\
 & \\
 & \\
 & \\
\end{tabular}
}
\end{tabular}
}
 \caption{The digraphs $G_1$, $G_2^\circ$, $G_E$, and the multidigraph $MG_2$ associated to the three-layer cascade network. It is easy to see that there are no parallel edges between different nodes of $MG_2$, there is a single directed path between any two nodes of each connected component of $G_2^\circ$, and $G_E$ has no directed cycles.}\label{fig:3-layerMESSI}
 \end{figure}
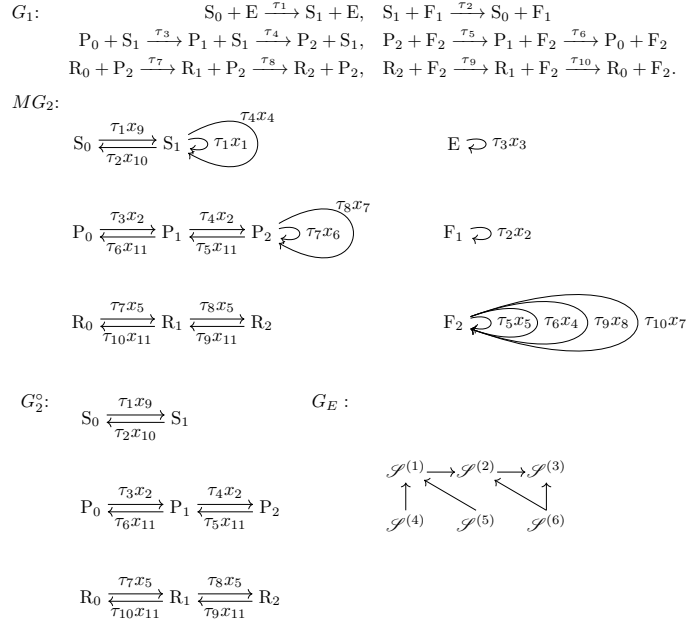
It is straightforward to check that the reduced network (without all intermediate species) associated to the three-layer cascade satisfies the hypotheses of Theorems~\ref{th:ylebn} and~\ref{th:monoy} below.  
\end{example}

Our next theorems give simple combinatorial conditions to ensure monostationarity and the lebn property of core MESSI networks. 

We first recall that a {\em directed cactus graph} is a strongly connected digraph in which each edge is contained in exactly one directed cycle or, equivalently, there is a single directed path between any two nodes~\cite{Cact1,Cact2}.

\medskip

\begin{theorem}~\label{th:ylebn} Let $G$ be the underlying digraph of a MESSI system without intermediates. 
Consider a minimal partition of the set of species as in \eqref{eq:partition} and the associated digraph $G_2^0$ from Definition~\ref{def:digraphs}.

\smallskip

Assume that the associated multidigraph $MG_2$ does not have parallel edges between different nodes.  If each connected component of $G_2^0$ is a cactus digraph, then $G$ and any non-confluent extension $G_C$ of $G$ are linearly equivalent to a binomial network.

\end{theorem}

We next show that similar combinatorial conditions ensure monostationarity.

\smallskip

\begin{theorem}~\label{th:monoy} Let $G$ be the underlying digraph of a MESSI system without intermediates. 
Consider a minimal partition of the set of species as in  and the associated digraphs $G_2^0$ and $G_E$.  If $G_E$ has no directed cycles,  and the assumptions on $MG_2$ and $G_2^0$ in the statement of Theorem~\ref{th:ylebn} hold, then the system is monostationary.
\end{theorem}

\smallskip

The reduced networks of all the examples treated in Section~\ref{sec:ERK} are MESSI systems; furthermore, all of them but the three-layer cascade with feedback loop satisfy the hypotheses of Theorems~\ref{th:ylebn} and~\ref{th:monoy}.  In fact, the associated digraph $G_E$ for the three-layer cascade with feedback loop has a directed cycle while the others do not. 
Before we prove Theorems~\ref{th:ylebn} and~\ref{th:monoy}, we show the sets $\Sp^{(1)},\dots,\Sp^{(m)}$ and the graph $G_E$ for the examples we have seen previously. 

\begin{example}

For the Example \ref{ex:cascade2} (two-layer cascade) 
a minimal partition of the set of species is:
$$\Sp^{(1)}=\{S_0,S_1\},\ 
\Sp^{(2)}=\{P_0,P_1\},\ 
\Sp^{(3)}=\{E\},\ 
\Sp^{(4)}=\{F\}\ $$
and the graph $G_E$ is

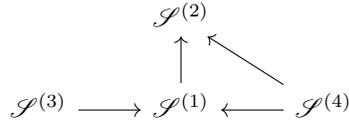
\begin{center}
\begin{tikzcd}
& \Sp^{(2)} & \\
\Sp^{(3)}\arrow[r]& \Sp^{(1)}\arrow[u] & \Sp^{(4)}\arrow[l]\arrow[ul]
\end{tikzcd}
\end{center}

For the Example \ref{ex:3layercascadeloop} (three-layer cascade with feedback loop) with three distinct enzymes we can consider the minimal partition $\Sp^{(1)}=\{S_0,S_1,S_2\},\ 
\Sp^{(2)}=\{P_0,P_1,P_2\},\ 
\Sp^{(3)}=\{R_0,R_1,R_2\},\Sp^{(4)}=\{E\},\
\Sp^{(5)}=\{G\},\
\Sp^{(6)}=\{F_1\},\ 
\Sp^{(7)}=\{F_2\},\ 
\Sp^{(8)}=\{F_3\} $
and the graph $G_E$ is

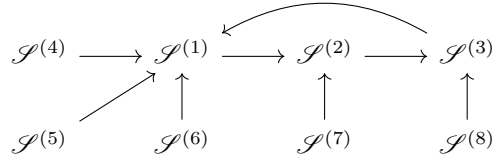
\begin{center}
\begin{tikzcd}
\Sp^{(4)}\arrow[r]&\Sp^{(1)}\arrow[r]&\Sp^{(2)}\arrow[r]&
\Sp^{(3)}\arrow[ll,bend right]\\
\Sp^{(5)}\arrow[ur]&
\Sp^{(6)}\arrow[u]&
\Sp^{(7)}\arrow[u]&
\Sp^{(8)}\arrow[u]&
\end{tikzcd}
\end{center}

For distinct configurations of enzymes the only differences are in $\Sp^{(6)},\Sp^{(7)},\Sp^{(8)}$ and the edges coming from them. For instance, if $F_2=F_3$ them  $\Sp^{(7)}=\Sp^{(8)}$ and there is two edges coming from it, going to $\Sp^{(2)}$ and $\Sp^{(3)}.$

Finally, for the $n$-site sequential distributive phosphorylation network in Figure~\ref{fig:nsite},  take 
$\Sp^{(1)}=\{S_0,S_1,\dots,S_n\},\ 
\Sp^{(2)}=\{E\},\ 
\Sp^{(3)}=\{F\}\ $
and the graph $G_E$ is

\begin{center}
\begin{tikzcd}
\Sp^{(3)}\arrow[r]& \Sp^{(1)} & \Sp^{(2)}\arrow[l].
\end{tikzcd}
\end{center}

\end{example}

When $G$ is a directed cactus graph, we can consider the tree structure over the underlying undirected graph $\widetilde{G}$ of $G$. For this purpose, the set of vertices of $\widetilde{G}$ is partitioned into three subsets: the set of vertices of degree two that are included in exactly one cycle; the vertices that do not belong to any cycle; and the remaining vertices, also called \emph{hinges}, which belong to at least one cycle. The tree representation $T_{\widetilde{G}}$ of $\widetilde{G}$ is obtained by keeping the vertices in the last two subsets and their corresponding edges, and by replacing all the cycles in $\widetilde{G}$ with cycle nodes and adding an edge from a hinge node to a cycle node if the hinge vertex belongs to the corresponding cycle in $\widetilde{G}$.

\begin{lemma}\label{lem:incidence}
 Let $G$ be a directed cactus graph with $n$ vertices and $r$ edges, and let $I_G$ be the $n\times r$ incidence matrix of the graph $G$. There exists an invertible matrix $K\in\Q^{n\times n}$ such that the first $n-1$ rows of $K\cdot I_G$ have only two nonzero entries, both with different signs. 
\end{lemma}

\begin{proof}
 Consider a leaf of the tree representation $T_{\widetilde{G}}$ of $\widetilde{G}$ (that is, a vertex of degree one). This leaf is associated to a directed cycle in $G$ with $n_1$ vertices ($n_1\geq 2$) and assume without loss of generality that the vertices in $G$ are numbered so that the first $n_1$ nodes and $n_1$ edges correspond to this cycle ($1\to 2 \to \dots \to n_1\to 1)$. If there is a hinge, assume $n_1$ is the hinge. Then the first $n_1$ columns of $I_G$, the incidence matrix of $G$ correspond to this cycle, and the last $r-n_1$ columns have zeros in the first $n_1-1$ entries:
 {\small 
 \[
  I_G = \left (
\begin{array}{ccccc|c|ccc}
-1 & 0 & \dots & \dots& 0 & 1 &  & & \\
1 & -1 & \ddots & & \vdots & \vdots & &  & \\
0 & 1 & \ddots & & 0 & & & 0 & \\
\vdots & & \ddots & \ddots &0 & \vdots & &  & \\
0 & \dots&  &1 & -1 & 0 & &  & \\ \hline
0 & \dots&  & 0 & 1 & -1 & {I_G}_{n_1,n_1+1}& \dots & {I_G}_{n_1,r}\\ \hline
&  &  & &  & 0 & &  & \\ 
 &  & 0 & &  & \vdots & \multicolumn{3}{c}{{I_G}_{(n_1+1\dots r),(n_1+1\dots r)}} \\ 
\makebox[0pt][l]{$\smash{\underbrace{\phantom{%
    \begin{matrix} -1 & -1 & \ddots & 0 & -1 & -1\end{matrix}}}_{\text{$n_1$}}}$} & &  &  &  &  0 &  & &  \\
\end{array}
\right ).
 \vspace{5mm}
 \]
}
 
 We can consider the invertible $n\times n$ block matrix $K_1$ as follows: 
 {\small 
 \[
  K_1=\left (
\begin{array}{ccccc|ccc}
-1 & 0 & \dots &  & 0 &  & & \\
-1 & -1 & \ddots & & \vdots & & 0 & \\
\vdots & & \ddots & \ddots & \vdots& &  & \\
-1 & \dots&  &-1 & 0 & &  & \\
1 & \dots&  &1 & 1 & &  & \\ \hline
 & &  &  & & &  & \\
 & & 0 &  & & & Id_{n-n_1} & \\ 
\makebox[0pt][l]{$\smash{\underbrace{\phantom{%
    \begin{matrix}-1 & -1 & \dots &-1 & 0\end{matrix}}}_{\text{$n_1$}}}$} &  &  &  &  &  & &  \\
\end{array}
\right ).
\vspace{5mm}
 \]
 }
 
By multiplying $K_1$ and $I_G$ we obtain the block matrix
{\small 
\begin{equation}\label{eq:incidence}
   K_1I_G=\left (
\begin{array}{ccc|c|ccc}
& & & -1 &  & & \\
 & Id_{n_1-1} & &\vdots & & 0 & \\
 & &  & -1 & &  & \\ \hline
0 & \dots&   0 & 0 & {I_G}_{n_1,n_1+1}& \dots & {I_G}_{n_1,n}\\ \hline
 &  & & 0 & &  & \\ 
 &  0 &  & \vdots & \multicolumn{3}{c}{{I_G}_{(n_1+1\dots r),(n_1+1\dots,r)} } \\ 
\makebox[0pt][l]{$\smash{\underbrace{\phantom{%
    \begin{matrix} 1  & \ddots & \ddots & 1 & -1\end{matrix}}}_{\text{$n_1$}}}$}  &  &  &  0 &  & &  \\
\end{array}
\right )
\vspace{5mm}
\end{equation}}
whose last block equals the incidence matrix of the subgraph of $G$ obtained by removing the first $n_1-1$ vertices and the edges from the corresponding cycle. By an inductive argument we see that we can obtain finitely many invertible matrices $K_1,K_2\dots$ such that the product gives an invertible matrix $K\in\Q^{n\times n}$ that yields the desired result. 
\end{proof}

We can now prove Theorems~\ref{th:ylebn} and~\ref{th:monoy}.

\begin{proof}[Proof of Theorem~\ref{th:ylebn}]  
 We start by proving that $G$ is lebn. As we pointed out in Remark~\ref{rmk:eq_from_G2^0}, as there are no intermediate species, the equations of the mass-action system determined by $G$ can be reconstructed from $G_2^0$ in a mass-action fashion by treating the labels of the edges as if they were reaction constants. As the associated multidigraph $MG_2$ does not have parallel edges between different nodes, these labels are monomials in the species concentrations. 
 
 Consider $1\leq \alpha\leq d$. By the assumption of minimality of the partition, there is a connected component of $G_2^0$, which we denote by $H_\alpha$, with vertices the species in $\Sp^{(\alpha)}$. Consider, without loss of generality, that $\Sp^{(\alpha)}$ consists of the first $n_\alpha$ species, i.e. $\Sp^{(\alpha)}=\{X_1,\dots,X_{n_\alpha}\}$. By Equation~\eqref{eq:CRNN} and Remark~\ref{rmk:eq_from_G2^0}, we have
 \[
  (f_1,\dots,f_{n_\alpha})^T \, = \, N_\alpha \cdot R_{\kappa,\alpha}(x),
 \]
where we order the $r_\alpha$ edges in  $H_\alpha$, $R_{\kappa,\alpha}(x)$ is the vector of size $r_\alpha$ where if the $i$-th edge is $X_k \overset{\kappa_i x^{\gamma_i}}{\longrightarrow} X_\ell$ then $R_{\kappa,\alpha}(x)_i = \kappa_i x_kx^{\gamma_i}$ and $N_\alpha$ is the stoichiometric matrix which in this linear case coincides with the incidence matrix of $H_\alpha$. Then, by multiplying $(f_1,\dots,f_{n_\alpha})^T $ on the left by the invertible matrix $K_\alpha$ from Lemma~\ref{lem:incidence}, we obtain binomials in the monomials of $R_{\kappa,\alpha}(x)$. 

By repeating the reasoning above with each $1\leq \alpha\leq d$, we obtain a block matrix $M\in\Q^{s\times s}$ that yields the binomial equivalence for $(f_1,\dots,f_s)$. 
To end the proof, it is immediate to see from Proposition~\ref{prop:lebn} that any non-confluent extension $G_C$ of $G$ is also lebn.
\end{proof}

\begin{proof}[Proof of Theorem~\ref{th:monoy}]  
Assume that $G_E$ has no directed cycles. We define the following subsets of indices, as in the proof of Theorem~3.15 in \cite{messi}:
\begin{align*}
L_0=&\{\beta \geq 1 : \text{indegree} \text{ of }\Sp^{(\beta)}\text{ is }0\},\ \text{and for } k\geq 1:\\
L_k=&\{\beta \geq 1: \text{for any }  \Sp^{(\gamma)}\to\Sp^{(\beta)}\text{ in }G_E \text{ we have }
\gamma \in L_t \text{ with }  t<k \}\backslash \underset{t=0}{\overset{k-1}{\bigcup}} L_t.
\end{align*}

It is important to note that $L_0\neq \emptyset$ because the graph $G_E$ has no directed cycles.
For each $\alpha\geq 1$, fix $X_{i_\alpha}\in\Sp^{(\alpha)}$. Because of the minimality of the partition,  any other $X_i \in \Sp^{(\alpha)}$ lies in the connected component $H_\alpha$ of $G_2^0$ containing $X_{i_\alpha}$. 

Choose $X_{i_{1}},\dots,X_{i_{d}}$ species, with $X_{i_{\alpha}}\in\Sp^{(\alpha)}$, for each $\alpha=1,\dots,d$.
Take any other species $X_{i}\in\Sp^{(\alpha)}$ for some $\alpha\in L_k$, $X_{i}\neq X_{i_{\alpha}}$, with $L_k$ as above. 
As $H_\alpha$ is strongly connected, let $\rho(G)$ be the generator of the kernel of $\mathcal{L}(H_\alpha)$ as in~\eqref{eq:rhoi}. Then we can write $x_i=\frac{\rho_i(H_\alpha)}{\rho_{i_\alpha}(H_\alpha)}x_{i_\alpha}$. But as $H_\alpha$ is a directed cactus graph, there is a unique $j$-tree for each node $j$ in $H_\alpha$ and then each entry $\rho_j(H_\alpha)$ is a monomial in the variables $x_{i_{\beta}}$ with ${\beta}\in L_t$, with $t<k$. Hence, the concentration of $X_{i}$ can be expressed in terms of $x_{i_1},\dots,x_{i_d}$ in the form:
\begin{equation}\label{concentrationofcore}
x_{i}=\phi(\tau)\, x_{i_{\alpha}}\, x^{a},
\end{equation}
for some $\phi(\tau)\in\Q(\tau)$, where $x^{a}$ is a Laurent monomial that depends only on variables $x_{i_{\beta}}$ with ${\beta}\in L_t$, with $t<k$. Note that, if $k=0$, then $x^{a}=1$.

Suppose that the species of $G$ are ordered in the following way: first we put the species in $\Sp^{(\beta)}$,  for all $\beta$ such that $\beta\in L_0$, then the species in $\Sp^{(\beta)}$, such that $\beta\in L_1$, and so on. With this order, from Equation~\eqref{concentrationofcore} we obtain Laurent monomials of the form $x_{i_{\alpha}}\, x^{a}\, x_i^{-1}$ with $i\neq i_\alpha$, $\alpha\in L_k$ and $x^a$ involves variables $x_j$ with $j<i$ for any $i$ such that $X_i\in \Sp^{(\alpha)}$.  From these monomials we build the $d\times s$ matrix $A$ as in Theorem~\ref{th:5.3}.

On the other hand, there are $d$ independent conservation relations by 
the hypotheses of the statement and Theorem~3.2
in \cite{messi}. They are:
\begin{equation}\label{eq:consalpha}
\ell_{\alpha}(x)\, = \sum_{X_j \in \Sp^{(\alpha)}} x_j, \quad \alpha=1,\dots,d.
\end{equation}
Furthermore, $\dim(\mathcal{S}^{\perp})=d$, where $\mathcal{S}$ is the stoichiometric subspace.

Consider the conservation-law matrix $W$ according to the conservation laws \eqref{eq:consalpha}. 
Note that by hypotheses and Proposition~$5.6$ in \cite{messi}, ${\rm rank}(W) + {\rm rank}({\rm Exp})=s$. We will now apply Theorem~\ref{th:5.3}. We note first that if we consider a submatrix of $W$ with two columns corresponding to species in the same set $\Sp^{(\beta)}$, then, its determinant is zero, because the two columns are equal. So we are interested in submatrices with columns corresponding to species in different sets $\Sp^{(\beta)}$. Then, suppose that we choose the set $I=\{j_1,\dots,j_d\}$ such that the species $X_{j_{\beta}}\in\Sp^{(\beta)}$, for each $\beta=1,\dots,d$. We then have $W_I=Id_d$ and $A_I$ is an upper triangular matrix with ones on its diagonal entries.
Then, $\det(W_I)\det(A_{I})=1$ in all these cases, and zero in the other cases, as we wanted to prove.
\end{proof}

\section*{Acknowledgments} 

We are very thankful to Eugenia Ellis and Andrea Solotar for organizing the excellent project ``Matem\'aticas en el Cono Sur'', which lead to this work. We also thank Eugenia Ellis for the warm hospitality in Montevideo, Uruguay, on December 3-7, 2018, where we devised our main results. AD, MG and MPM were partially supported by UBACYT 20020220200166BA and CONICET PIP 11220200100182CO, Argentina. Rick Richter was partially supported by FAPEMIG RED-00133-21. We also acknowledge partial support from the PUE grant 22920170100037 of the Instituto de Investigaciones Matem\'aticas Luis A. Santal\'o, that covered a visit of Rick Rischter to Buenos Aires to end this project.

    \end{document}